%% file: ht-av-review.tex
\documentclass[11pt, onecolumn, draftcls]{IEEEtran}
\include{preamble}
\allowdisplaybreaks

\begin{document}
\IEEEoverridecommandlockouts

\title{Universal Hashing for Information Theoretic Security}

\author{Himanshu Tyagi and Alexander Vardy}

\maketitle

{\renewcommand{\thefootnote}{}\footnotetext{
\noindent Himanshu Tyagi is with the Indian Institute of
Science, Bangalore 560012, India. Email:
htyagi@ece.iisc.ernet.in.

Alexander Vardy is with the University of California, San
Diego, La Jolla, CA 92093, USA.  Email: avardy@eng.ucsd.edu.
}} \renewcommand{\thefootnote}{\arabic{footnote}}
\setcounter{footnote}{0}

\begin{abstract}
The information theoretic approach to security entails
harnessing the correlated randomness available in nature to
establish security.  It uses tools from information theory
and coding and yields provable security, even against an
adversary with unbounded computational power.  However, the
feasibility of this approach in practice depends on the
development of efficiently implementable schemes.
In this article, we review a special class of practical
schemes for information theoretic security that are based on
$2$-{\it universal hash families}.
Specific cases of secret key agreement and wiretap coding
are considered, and general themes are identified.
The scheme presented for wiretap coding is {\it modular} and
can be implemented easily by including an extra
pre-processing layer over the existing transmission codes.
\end{abstract}

\section{Introduction}
Random variations in physical observations constitute a
valuable resource for facilitating security in engineering
systems. Authentication keys can be extracted from noisy
recordings of biometric signatures
\cite{RatConBol01,JaiRosPan06}; unique signatures for
hardware devices can be generated by implanting a {\it
  physically uncloneable function} (PUF), implemented using
random manufacturing variations in the period of a
ring-oscillator \cite{Pap01, GasClaDijDev02}; secret keys
extracted from the random fade of a wireless communication
channel can be used for cryptographic applications
\cite{YeMRSTM10}; various physical layer security techniques
can be used to mitigate the security threats in cyberphysical
systems and ad-hoc networks \cite{WangL13, JainTB12}; and
wiretap codes can be used for protection against side-channel
attacks~\cite{BringerCL12}.  The information theoretic
approach for security entails developing a systematic theory
for designing and analyzing security primitives
based on harnessing physical randomness.  In this
approach, we treat physical observations as a source of
randomness hidden from the attacker and study the design of
optimal codes for accomplishing specific security
objectives.  One limitation of this approach is the
assumption that the attacker does not have a complete access
to or cannot manipulate the correlated randomness
 used for
implementing security. In lieu, we can provide information
theoretic security guarantees which hold even when the
attacker has unlimited computational power.

The origin of information theoretic security, as well as of
theoretical cryptography, lies in the seminal paper of
Shannon \cite{Sha49}.  This paper shows that in order to
securely transmit an $m$-bit random message over an insecure
public channel the transmitter and the receiver must share
an $m$-bit {\it perfect secret key}\footnote{Shannon
  \cite{Sha49} established the necessity of an $m$-bit
  perfect secret key only for the case when a one-time-pad
  is used for encryption. The general necessary condition
  for any scheme was shown in \cite{Mas88} (see, also,
  \cite[Problems 2.12 and 2.13]{KatLin07},
  \cite{HoCU11},\cite[Section VI]{TyaWat14ii}).}, $i.e.$, $m$
uniformly distributed bits that are concealed from an
eavesdropper with access to the public channel. This
requirement of large perfect secret keys is impractical, and
thus, the result of \cite{Sha49} is largely considered a
negative result.  Following the pioneering work of Diffie
and Hellman \cite{DiffieH76}, modern cryptography
circumvents this restriction by relaxing the security
requirement from information theoretic security to security
against a computationally bounded adversary.  However, a
different remedy is possible in situations where the
transmitter and the receiver have access to correlated
randomness, which is available to the eavesdropper only in
part.  Specifically, it was shown by Wyner in \cite{Wyn75ii}
that if the eavesdropper can access only a noisy version of
the observations of the legitimate receiver, secure
transmission\footnote{Strictly speaking, we are concerned
  with the transmission of confidential messages in the
  presence of passive eavesdroppers.} is feasible without
requiring any additional resources. Furthermore, it was
shown in \cite{BenBraRob88, Mau93,AhlCsi93} that information
theoretically secure secret keys can be extracted from
correlated random observations by communicating over an
insecure, public communication channel.  These works
constitute the basic foundations of information theoretic
security, suggesting that the requirement of large secret
keys for the feasibility of information theoretic security
can be circumvented if correlated randomness is available.

Inspired by these results, practical schemes for information
theoretically secure message transmission and secret key
agreement have been proposed, utilizing the correlated
randomness available in the physical communication channel
($cf.$~\cite{AonoHOKS05,MathurTMY08, CroftPK10, XiaoGT10})
or the correlated randomness extracted from physical
observations ($cf.$~\cite{RatConBol01,JaiRosPan06, Pap01,
  GasClaDijDev02}).
However, most of the practical schemes proposed have either
no theoretical guarantees of performance or are suboptimal.
In fact, even for the basic problems of {\it secret key
  agreement} and coding for a {\it wiretap channel}, optimal
practical schemes are few and have emerged only over the
last decade (see~\cite{BlochTMM06, DodOstReySmi08,YeNar12,
  RenesRS13, ChouBA13} for optimal schemes for secret key
agreement and the review article \cite{HarrisonABMB13} for
references on optimal codes for a wiretap channel).

In this article, we review a class of practical coding
schemes for attaining information theoretically secure
secret key agreement as well as for information
theoretically secure message transmission in a wiretap
channel model.  Specifically, we focus on schemes that use
{\em $2$-universal hash families} (UHF) \cite{CarWeg79} (see
Section~\ref{s:2universal} for the definition of a UHF) as a
building block. This restriction in scope is for two
reasons: First, UHFs are easy to implement and are ideally
suited for lightweight cryptography ($cf.$~\cite{Kra94,
  YukselKS04}), and second, while review articles are
available that cover the role of error-correcting codes in
physical layer security ($cf.$ \cite{MukherjeeFHS14,
  HarrisonABMB13}), the UHF based schemes for wiretap 
channels are recent and are not well-known.

The remainder of this article is organized as follows. We
begin by describing the secret key agreement and the wiretap
coding problem in the next section. In the subsequent
section, we define a UHF and discuss its basic properties
and some practical implementations.  In the final two
sections, we review UHF based coding schemes for secret key
agreement and wiretap channels.

\section{Primitives for Information theoretic security}
In this section, we describe two basic primitives for
information theoretic security.  Both rely on the
correlation in the random observations of legitimate
parties; however, the form of correlation is different in
each.  The first of these, namely secret key agreement, is
concerned with extracting shared secret bits from noisy
correlated random data.  The second, coding for wiretap
channels, focuses on sending data over a noisy channel when
a passive eavesdropper observes noisy versions of the
transmissions. The two problems seem to be different in
their scope and objective. Yet similar schemes based on
error-correcting codes and UHF will be seen to be optimal
for both in many cases.

Note that the basic cryptographic primitives of oblivious
transfer \cite{Rab81} and bit commitment \cite{Blu83}, too,
have information theoretically secure counterparts; see, for
instance, \cite{CreKil88, NasWin08, WinWul12, AhlCsi13,
  RaoPra14, TyaWat14ii} and
\cite{WinNasIma03,ImaMorNasWin06,RanTapWinWul11,
  TyaWat14ii}, respectively, for treatments of information
theoretically secure oblivious transfer and bit commitment.
However, there are only a few practical schemes available
($cf.$~\cite{ImaMorNasWin06}), and they will not be reviewed
here.

\subsection{Secret key agreement}
Discrete, correlated random variables $X$ and $Y$, with
arbitrary but known distribution $\bPP{XY}$, are observed by
the first and the second party, respectively.
The parties seek to agree on random, unbiased bits. These
correlated random variables correspond to random physical
observations and can be derived, for instance, from
different noisy recordings of the same biometric
fingerprint, or from the random fade observed in a wireless
communication channel. The parties also have access to a
public communication channel such as a shared public server,
or a broadcast channel, or any other insecure communication
network. They can use this communication channel to exchange
bits with each other; however, the bits exchanged will be
available to a (passive) eavesdropper. The mode of
communication allowed depends on the application at
hand. For instance, in the biometric and PUF applications,
only one sided communication from $X$ to $Y$ is available
since $Y$ corresponds to a later (in time) recording of $X$
itself. In general, the parties can execute an interactive
communication protocol $\Pi$ with multiple rounds of
interaction and possibly randomized communication in each
round\footnote{The communicated data $\Pi$ is sometimes
  referred to as {\em helper data}.}.  The goal is to derive
a secret key $K$ consisting of bits $(K_1, ..., K_l)$ such
that (i) with large probability, both parties can recover
$K$ accurately; (ii) bits $(K_1, ..., K_l)$ are almost
independent and unbiased; and (iii) an eavesdropper with
access to the communication $\Pi$ and a side information $Z$
cannot ascertain any information about $K$.
\begin{figure}[ht]
\begin{subfigure}[b]{\linewidth}
\centering \includegraphics[scale=.35]{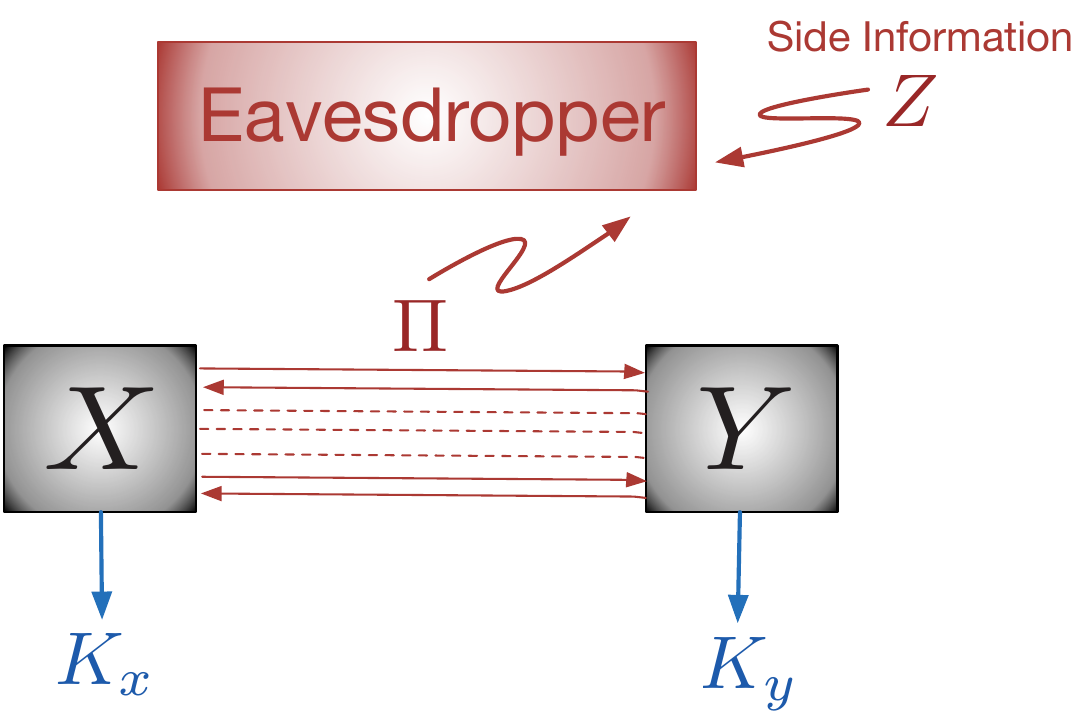}
\caption{Secret key agreement protocol}
\end{subfigure}
\begin{subfigure}[b]{.5\linewidth}
\centering \includegraphics[scale=.25]{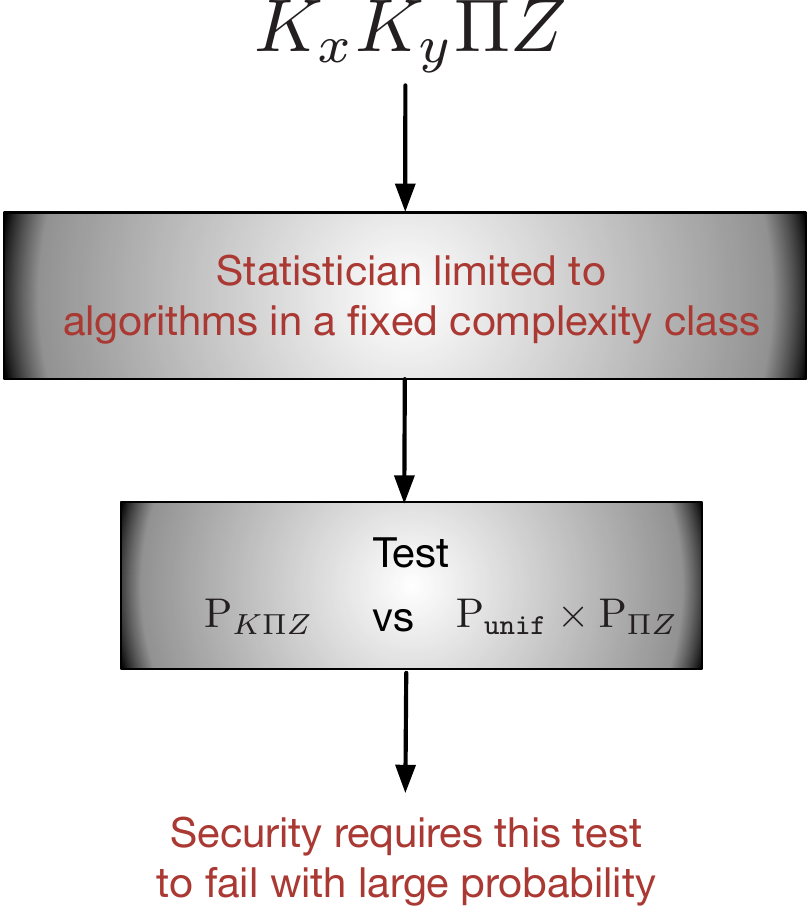}
\caption{Computational security requirement}
\end{subfigure}
\begin{subfigure}[b]{.5\linewidth}
\centering \includegraphics[scale=.25]{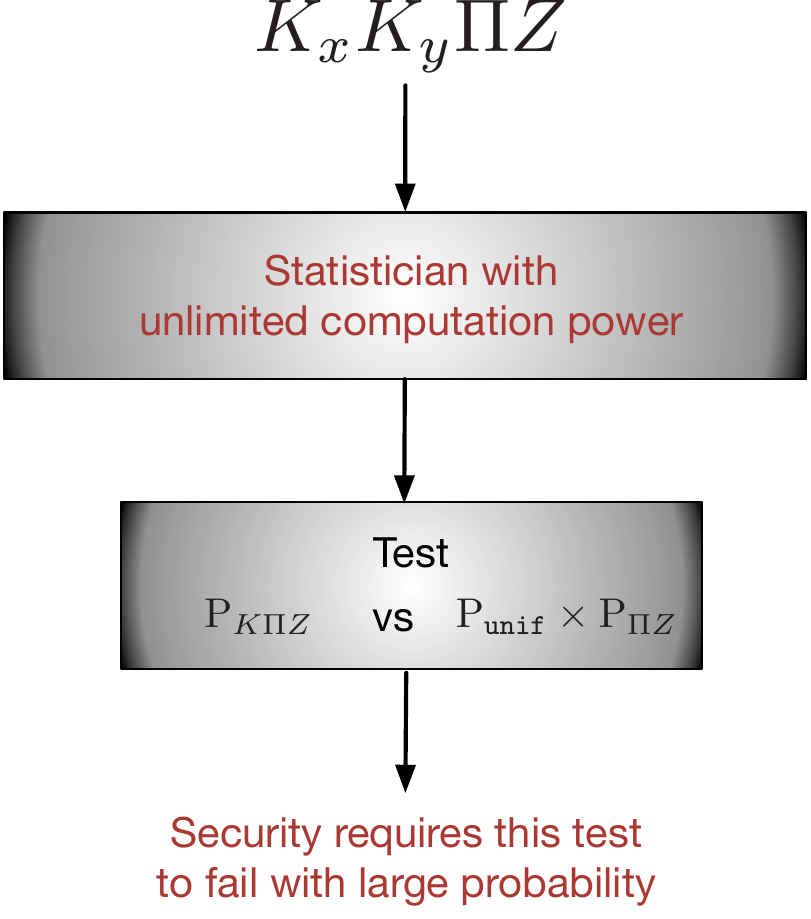}
\caption{Information theoretic security requirement}
\end{subfigure}
\caption{Illustration of secret key agreement}
\label{f:SK}
\end{figure}

Condition (i) above constitutes the {\it recoverability}
requirement. Parties must form estimates $K_x$ and $K_y$ of
$K$ such that
\[
\bPr{K_x = K_y = K} \geq 1-\ep,
\]
for a suitably small parameter $\ep$.

Conditions (ii) and (iii) above constitute the {\it
  security} requirement. Traditional notion of cryptographic
security is computational and requires ($cf.$
\cite{GolMic84}) that a computationally bounded
adversary with access to efficient
algorithms for solving problems in a particular complexity
class, but not beyond it, cannot reliably distinguish if the
observed outputs of the secret key agreement protocol $(K,
\Pi, Z)$ are coming from the real protocol or an ideal one
with all values of $K$ equally likely for each realization
of $(\Pi, Z)$. In contrast, \cite{BenBraRob88,
  Mau93,AhlCsi93} initiated the study of the secret key
agreement problem under {\it information theoretic security}
where the computationally bounded adversary above is
replaced by an unrestricted one with access to any statistical 
test\footnote{For another
  connection between binary hypothesis testing and secret
  key agreement, see \cite{TyaWat14,
    TyaWat14ii}.}. Formally, it is required that the {\em
  statistical distance} between the joint distribution
$\bPP{K\Pi Z}$ and $\bPP{\mathtt{unif}}\times \bPP{\Pi Z}$
is small. Two popular measures of statistical distance that
have been used in secret key agreement literature are the
K-L divergence \cite{Mau93, AhlCsi93, BenBraCreMau95, Csi96,
  CsiNar04, CsiNar08, Hay11}
\[
D(\dP\|\dQ) =\sum_i \bPP i \log \frac{\bPP i}{\bQQ i},
\]
and the total variation distance \cite{RenWol05, Ren05, Hay13}
\[
\|\bPP{} - \bQQ{}\|_1 = \frac 12 \sum_i |\bPP i - \bQQ i|.
\]
For concreteness, we shall consider security of secret keys
under the total variation distance and require
\[
\left\|\bPP{K\Pi Z} - \bPP{{\tt unif}}\times \bPP{\Pi
  Z}\right\|_1 \leq \delta,
\] 
where $\bPP{{\tt unif}}$ is a uniform distribution on
$l$-bits.  Figure~\ref{f:SK} illustrates the setup and a
comparison of the computational and information theoretic
security criteria.  For given values of recoverability and
security parameters $\ep$ and $\delta$, we seek to design
secret key agreement protocols that yield as many bits of
secret key $K$ as possible, $i.e.$, the largest possible
value of $l$ above.

The theoretical limits of the length of secret keys possible
have been studied extensively: \cite{Mau93} and
\cite{AhlCsi93} considered the case when the underlying
observations are {\it independent and identically
  distributed} (IID) and, under a weaker notion of security
than that above, characterized the {\it secret key
  capacity}, $i.e.$, the maximum rate of secret key length
per observation; \cite{BenBraCreMau95, Csi96, AhlCsi98,
  MauWol00} provide basic tools for attaining the stronger
notion of security above without any loss of performance;
\cite{CsiNar04} establishes the secret key capacity for a
multiparty version of the problem; \cite{RenWol05, Ren05}
derive bounds on the secret key length for the single-shot
case above, when only one-sided communication is allowed;
\cite{TyaWat14, HayTyaWat14i, HayTyaWat14ii} give the
best-known bounds for the general problem above, stressing
on the role of interactive communication. However, none of
these works give an efficient secret key agreement
scheme. The literature on constructive coding schemes, on
the other hand, is narrow and has focused mostly on the case
with one-sided communication. In this article, 
we will discuss a class of
constructive schemes for secret key agreement that rely on UHFs.

\subsection{Coding for wiretap channel}
The problem of wiretap coding is that of 
transmitting a message with confidentiality from an
eavesdropper with side-information. 
Specifically, a senders seeks to communicate a message $M$
to a receiver by using transmissions over a noisy
communication channel $T$ with inputs from a set $\cX$ and
outputs from a set $\cY$. For each input $x$ to $T$, the
receiver observes an output $y$ with a given probability
density $T(y|x)$. Furthermore, for each transmission $x$ an
eavesdropper observes the output $z$ of another
communication channel $W$. It is required that while the
legitimate receiver decodes $M$ with a low probability of
error, while the message remains concealed from the eavesdropper
(or the {\em wire-tapper}).  See Figure~\ref{f:wiretap} for
an illustration.
\begin{figure}[ht]
\centering \includegraphics[scale=.4]{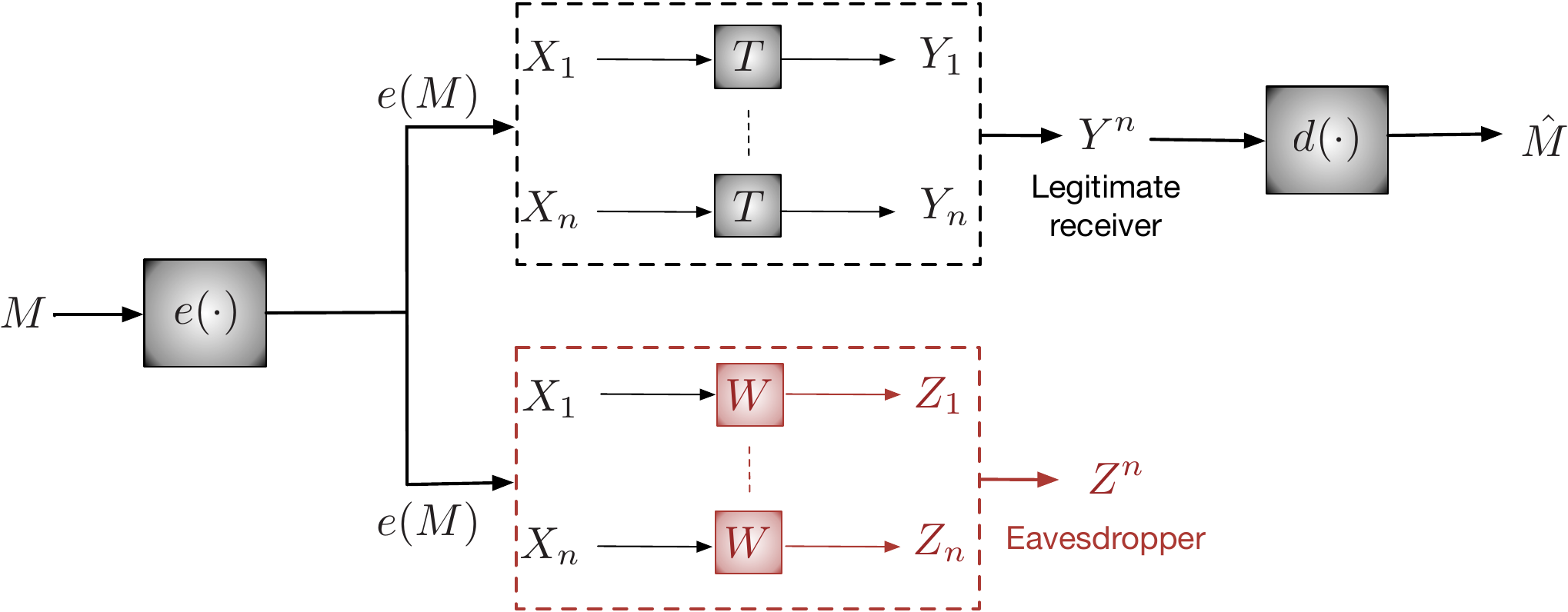}
\caption{Illustration of wiretap channel}
\label{f:wiretap}
\end{figure}

An $(n,k)$ code for this wiretap channel consists of a
(stochastic) encoder $e: \{0,1\}^k\rightarrow \cX^n$ and a
decoder $d: \cY^n \rightarrow \{0,1\}^k$.  A random message
$M$ is sent as $e(M)$ and decoded as $\hat M = d(Y^n)$,
where $Y^n = (Y_1,..., Y_n)$ denotes the outputs for $n$
independent uses of the channel $T$ for inputs $(X_1, ...,
X_n) = e(M)$. At the same time, an eavesdropper gets to
observe the outputs $Z^n$ corresponding to transmitting the
inputs $X^n$ over the channel $W$.  It is required that the
code $(e, d)$ ensures high reliability, $i.e.$ $\bPr{M \neq
  \hat M} \approx 0$ (it is required that $\bPr{M \neq \hat
  M}$ goes to $0$ sufficiently rapidly in $n$), and ensures
security under an appropriate notion.  The rate of this code
is $(k/n)$; the maximum possible asymptotic rate of a
wiretap code is called the {\it wiretap capacity} of $(T,
W)$.

This basic model was introduced by Wyner in~\cite{Wyn75ii}
where he considered a {\it degraded wiretap channel} where
$W = V\circ T$ for some stochastic mapping $V$, $i.e.$, the
eavesdropper's observation is a further noisy version of the
legitimate receiver's observation, and for an input $x$ the
eavesdropper's channel produces an output $z$ with
probability $W(z|x) = \sum_{y}V(z|y)T(y|x)$.  For this
important special case, Wyner characterized the wiretap
capacity under the {\it weak security} requirement
given by
\begin{align}
\lim_{n\rightarrow \infty}\frac 1 n I(M \wedge Z^n) = 0,
\nonumber
\end{align}
where the message $M$ is a uniform random variable and
$I(U\wedge V)$ denotes the mutual information between random
variables $U$ and $V$~\cite{CsiKor11}.  Later, Csisz\'ar and
K\"orner characterized $\wircap$ for all discrete,
memoryless wiretap channels \cite{CsiKor78}.

Interestingly, the wiretap capacity remains unchanged even
if we drop the normalization by $n$ in the weak security
condition above and require {\it strong security}
\cite{Csi96}
\begin{align}
\lim_{n\rightarrow \infty}I(M \wedge Z^n) = 0, \nonumber
\end{align}
for a uniform message $M$. A still more demanding notion of
security 
introduced\footnote{
This notion of security is termed {\it mutual information security} in 
\cite{BelTesVar12ii} and {\it source
  universality} in \cite{HayM13}.} 
in \cite{BelTesVar12ii}
requires security not only for a uniform message $M$ but any
random message $M$ and is given by
\begin{align}
\lim_{n\rightarrow \infty}\max_{\bPP{M}}I(M \wedge Z^n) =0.
\nonumber
\end{align}
In fact, \cite{BelTesVar12ii} extended the cryptographic
notion of {\it semantic security} ($cf.$~\cite{GolMic84})
to the wiretap channel and showed that it is implied by the 
security requirement above. In this article, we shall use
the term {\it semantic security} synonymously with the security requirement
above, keeping in mind that, in fact, we are demanding
something even stronger than semantic security.

It remains an open question if the wiretap capacity can be
achieved under semantic security, in general. However, for specific
wiretap channels, codes that achieve wiretap capacity while
ensuring semantic security have been proposed recently (
$cf.$ \cite{HayMat10, MahV11, BelTes12, BelTesVar12ii,
  LinLuzBelSte13}).  In particular, the schemes in
  \cite{HayMat10, Hay11, BelTes12, HayM13, Hay13,TyagiV14}
  rely on UHFs and are discussed below.  
\section{$2$-universal hash families}\label{s:2universal}
The key primitive that underlies all the schemes that will
be discussed in this article is a UHF. Universal hashing was
introduced by Carter and Wegman in their seminal work
\cite{CarWeg79} as a multipurpose tool for theoretical
computer science and was applied for privacy amplification
first in \cite{BenBraRob88}.  A UHF is, roughly speaking, a
family of functions such that the random mapping obtained by
uniformly choosing a function from this family is {\it
  almost} invertible. In information theory, as in
theoretical computer science, many of the proofs are
completed using a random mapping or binning or coloring of
elements of a set. It turns out that most of the tasks that
can be done using a completely random mapping can also be
done by a randomly selected member of a UHF.
Moreover, while implementing a
random mapping is not practical, structured implementations
of certain UHFs are available ($cf.$~\cite{Kra94,
  YukselKS04, HayT15} and \cite[Appendix II]{Hay13}).
Thus, UHFs constitute an efficiently 
implementable substitute for
random mappings.

Formally, a family $\cF$ consisting of mappings $f: \cX
\rightarrow \{1, ...,2^k\}$ is a ($k$-bit) UHF if
for every $x\neq x^\prime$
\begin{align}
\frac 1{|\cF|} |\{f\in \cF : f(x) = f(x^\prime)\}|\leq
2^{-k},
\label{e:UHF_property}
\end{align}
$i.e.$, the random mapping $F$ chosen uniformly over $\cF$
maps two distinct values to the same output with probability
less than $2^{-k}$.

The diverse applications of UHFs in information theoretic
security include: secret key agreement
($cf.$~\cite{BenBraCreMau95, RenWol05, Hay11}), quantum key
distribution ($cf.$~\cite{Ren05}), biometric and hardware
security ($cf.$~\cite{DodOstReySmi08}), and coding for
wiretap channels ($cf.$~\cite{Hay06, HayMat10, BelTes12,
  HayM13}); see \cite{Sti02} for other applications in
cryptography. In these applications, the importance of a UHF
lies in the role it plays in randomness extraction in source
and channel models. In a source model, we consider a randomness
which is observed by a legitimate party and is generated by a fixed
distribution. On the other hand, in a channel model, the
randomness is observed by an adversary and its distribution
is controlled by a legitimate party.
Basic results
were first derived for source models and, later, variants of
these basic results were derived for channel models; we
shall review the results for both these cases below.

\subsection{Source models}
In a source model, the available random observation
and eavesdropper's observation are modeled by correlated random
variables $(X,Z)$. In applications such as secret key
agreement, we seek to design a primitive that extracts from
$X$ uniformly distributed random bits that are almost
independent of $Z$. UHFs 
described above provide a
constructive tool for realizing such a primitive. First, 
we consider
the special case of a constant $Z$. The main result
here is the {\it leftover hash lemma} which shows roughly
that the output of a randomly chosen member of a $k$-bit UHF
applied to a random variable $X$ constitutes uniformly
random bits, provided that $k$ is smaller than a
threshold. Different versions of {\it leftover hash lemma}
are available in literature, each with a slightly different
choice of this threshold ($cf.$ \cite{ImpZuc89, ImpLevLub89,
  BenBraCreMau95, HastadILL99, Sti02, Ren05, RenWol05}).  We
review a version due to \cite{Ren05, RenWol05} where the
aforementioned threshold for randomness extraction is given
by the {\it smooth min-entropy} $H_{\min}^\ep(\bPP X)$ of the
underlying random variable $X$, defined as follows
\cite{Ren05, RenWol05}: The {\it min-entropy} of $X$ is
given by \cite{Ren61}
\[
H_{\min}(\bPP X) = \min_x -\log \bP X x,
\] 
and the $\ep$-smooth min-entropy of $X$ is defined as
\cite{RenWol04, RenWol05, Ren05}
\[
H_{\min}^\ep(\bPP X) = \sup_{\bQQ{}:\, \|\bPP{} - \bQQ{}\|_1 \leq
  \ep}\,\, H_{\min}(\bQQ{}).
\]
The leftover hash lemma uses a randomly selected member of a
given UHF. In order to facilitate this random selection, we
assume that a random seed $S$ distributed uniformly over a
discrete set $\cS$ is available to both the legitimate party
as well as the eavesdropper. While bounding the leaked
information of the extracted bits, eavesdropper's knowledge
of the random seed is taken into account as well.

\begin{lemma}[{{\bf Leftover hash: No side information}}]
Consider random variables $X$ taking values in a finite set
$\cX$. Then, for a $k$-bit UHF consisting of mappings $\{f_s,
s\in \cS\}$ and a random seed $S$ distributed uniformly over
the set $\cS$, it holds for every $\ep \in [0, 1)$ that
\[
\|\bPP {f_S(X)S} - \bPP {\mathtt{unif}}\times \bPP S\|_1
\leq \ep+ \frac 12 \sqrt{2^{k-H_{\min}^\ep(X)}}.
\]
\end{lemma}
The first instance of a variant of this result, for the
special case $\ep = 0$, appeared in~\cite{ImpLevLub89} (see,
also, \cite{HastadILL99} for further strengthening of this
result).  The term ``leftover hash lemma'' appeared in
\cite{ImpZuc89} where a strengthening of the result
of~\cite{ImpLevLub89} was given with R\'enyi entropy of
order $2$ in place of min-entropy.  The
form given above is a special case of a general result in
\cite{Ren05, RenWol05} for the case where, in addition to
the random seed $S$, the eavesdropper observes a
(possibly continuous-valued) random variable $Z$.  In this general
version, the threshold $H_{\min}^\ep(X)$ is replaced by the
$\ep$-{\it smooth conditional min-entropy} given by
\cite{RenWol04, Ren05}
\[
H_{\min}^\ep(\bPP{XZ}|Z) = \sup_{\bQQ{XZ}:\, \|\bPP{XZ} - \bQQ{XZ}\|_1
  \leq \ep}\,\, H_{\min}(\bQQ{XZ}| Z),
\]
where $H_{\min}(\bQQ{XZ} | Z)$ denotes the conditional
min-entropy
\[
H_{\min}(\bPP{XZ}| Z) = \sup_{\bQQ Z: \mathtt{supp}(\bPP Z) \subset
\mathtt{supp}(\bQQ Z)}H_{\min}(\bPP{XZ}| \bQQ Z)
\]
and, for $\bPP Z$ and $\bQQ Z$ with densities $f_{\bPP {}}$
and $f_{\bQQ {}}$ (with respect
to a
measure $\mu$ on $\cZ$), respectively,
\[
H_{\min}(\bPP{XZ}| \bQQ Z)
= \inf_{x\in \cX,z \in \mathtt{supp}(\bQQ Z)} - \log
\frac{\bP{X|Z}{x|z}f_{\bPP{}}(z)}{f_{\bQQ {}}(z)}.
\]
Note that smooth min-entropies replaces Shannon entropies as
a measure of randomness in the context of randomness
extraction (see \cite[Section VI]{BenBraCreMau95} for
further discussion). However, for IID observations $X^n$,
Shannon entropy constitutes the leading asymptotic term in
smooth min-entropy of $\bPP{X^n}$ ($cf.$~\cite{Ren05}).  We
depict the result of \cite{Ren05, RenWol05} in
Figure~\ref{f:leftover}. Below, we recall a further
generalization where the side information $Z$ available to
the eavesdropper consists of a finite-valued random variable
$Z_1$ and a continuous-valued random variable $Z_2$; see,
for instance,~\cite[Appendix B]{HayTyaWat14ii}) for a proof.
\begin{figure}[H]
\centering \includegraphics[scale=.4]{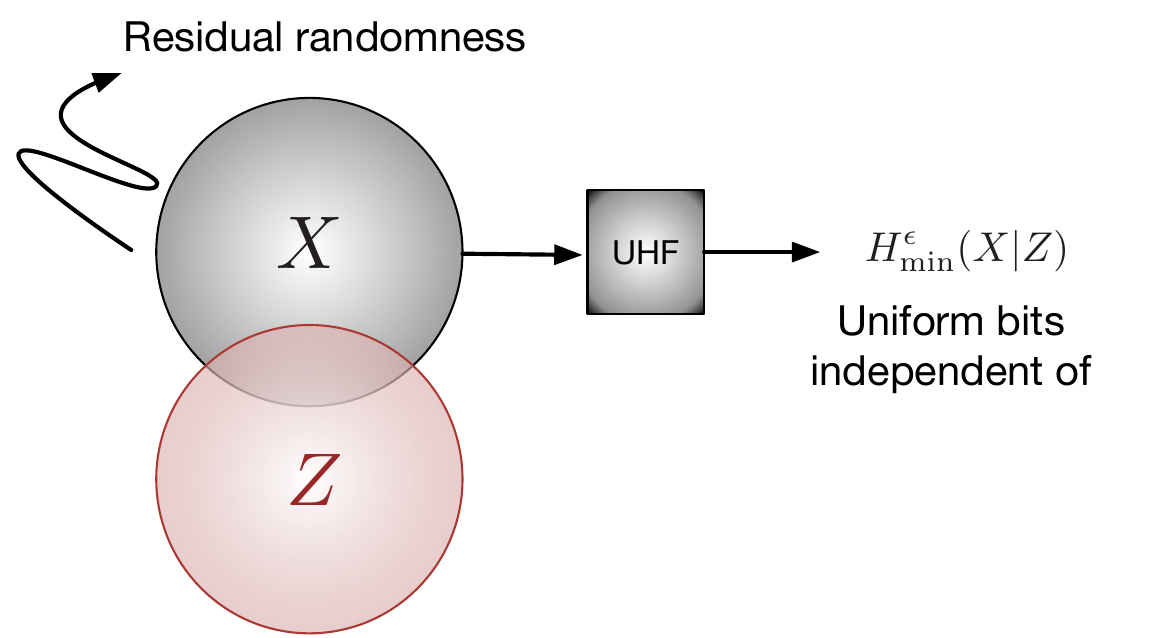}
\caption{Leftover hash property of UHFs}
\label{f:leftover}
\end{figure}

\begin{lemma}[{{\bf Leftover hash}}]\label{l:leftover_hash}
Consider random variables $X, Z_1, Z_2$ taking values,
respectively, in a finite set $X$, a (possibly uncountable) set $\cZ_1$, and
a finite set $\cZ_2$. Then, for a $k$-bit
UHF consisting of mappings $\{f_s, s\in \cS\}$ and a random
seed $S$ distributed uniformly over the set $\cS$, it holds
for every $\ep \in [0, 1)$ that
\[
\|\bPP {f_S(X)Z_1Z_2S} - \bPP {\mathtt{unif}}\times \bPP
       {Z_1Z_2S}\|_1 \leq\ep+ \frac 12
       \sqrt{|\cZ_2|2^{k-H_{\min}^\ep(\bPP{XZ_1}|Z_1)}}.
\]
\end{lemma}
In essence, the result above says that $H_{\min}^\ep(\bPP{XZ_1}|Z_1) - \log
|\cZ_2|$ almost uniform bits which are almost independent of
$(Z_1, Z_2)$ can be extracted from $X$. To measure
``almost'' uniformity and independence, the results above
use the total variation distance. An alternative form of the
leftover
hash lemma, with the K-L divergence replacing the variation
distance, was derived in \cite{BenBraCreMau95} and is
reviewed below.
\begin{lemma}[{{\bf Leftover hash: Divergence form}}]
Consider random variables $X, Z$ taking values,
respectively, in a finite set $X$ and a (possibly
uncountable) set $\cZ$. Then, for a $k$-bit UHF consisting
of mappings $\{f_s, s\in \cS\}$ and a random seed $S$
distributed uniformly over the set $\cS$, it holds that
\[
D\left(\bPP {f_S(X)ZS} \| \bPP {\mathtt{unif}}\times \bPP
{ZS}\right) \leq \frac{2^{k-H_{\min}(\bPP{XZ}|\bPP Z)}} {\ln
  2}.
\]
\end{lemma}
Note that by Pinsker's inequality ($cf.$~\cite{CsiKor11}),
the K-L divergence form yields the total variation distance form
(up to a constant factor).  On the other hand, using the
continuity of entropy in total variation distance, a K-L
divergence form was derived using the total variation distance
form above in \cite{Hay11} (see, also, \cite[Lemma
  1]{CsiNar04}).  Both the total variation distance form and the
K-L divergence form of the leftover hash lemma
given above combine the requirement of almost uniformity of
$f_S(X)$ and security of $f_S(X)$ from an observer of
$(Z,S)$ into a single criterion.  In fact, the result in
\cite{BenBraCreMau95} shows that $k - H(f_S(X)|ZS) = k -
H(f_S(X)) + I(f_S(X)\wedge ZS)$ is bounded above by
$2^{k-H_{\min}(\bPP{XZ}|\bPP Z)}/ {\ln 2}$, which in turn
implies that the mutual information $I(f_S(X)\wedge ZS)$ is
bounded above by the same quantity\footnote{The quantity $k
  - H(f_S(X)|ZS)$ was defined as a security index in
  \cite{CsiNar04} and was noted to equal $D\left(\bPP
       {f_S(X)ZS} \| \bPP {\mathtt{unif}}\times \bPP
       {ZS}\right)$.}.  In the information
theory literature, traditionally, mutual information has
been used as a measure of information
leakage\footnote{Bounds on leakage measured by R\'enyi
information quantities were derived recently in \cite{HayTan15}.}
($cf.$~\cite{Sha49, Wyn75ii, CsiKor78}), and the result above
says that the information about $f_S(X)$ leaked to the
eavesdropper is small as long as $k$ is sufficiently smaller
than $H_{\min}(\bPP{XZ}|\bPP Z)$. It was noted in \cite[Appendix
  III]{Hay06} that, under an almost uniformity assumption
for $f_S(X)$, a bound on $I(f_S(X)\wedge ZS)$ yields a bound
on the total variation distance $\|\bPP {f_S(X)ZS} - \bPP
{\mathtt{unif}}\times \bPP {ZS}\|_1 $.  On the other hand, a
counterexample was given to show
that a small $\|\bPP {f_S(X)ZS} - \bPP {\mathtt{unif}}\times
\bPP {ZS}\|_1 $ need not guarantee a small $I(f_S(X)\wedge
ZS)$.

In practice, one is interested in characterizing the optimal
tradeoff between information leakage and the range-size $k$
of the UHF used. For the case of IID observations $X^n$,
\cite{Hay08} considered the optimal $k = k_n(\ep)$ required
to attain a given leakage $\ep$ as a function of $n$ and
studied the second-order asymptotic term, for both the
total variation distance and the K-L divergence criteria (see the
textbook \cite{Han03} and the references therein for a
treatment of general sources beyond IID).  In a different
regime, \cite{Hay11} studied the exponential decrease in the
leakage for increasing $n$, for a fixed rate $k = nR$ for
the total variation distance based leakage and the mutual
information leakage, with a focus on the latter; optimal
exponents for decay rate of the total variation distance based
leakage as a function of $n$ were obtained in \cite{Hay13}.

\subsection{Channel models}
Another class of models relevant for the wiretap channel
entails
a channel $V:\cX \rightarrow \cZ$ between the legitimate
party and the eavesdropper. For each input $x \in \cX$
selected by the
legitimate party, the eavesdropper observes a random
variable $Z\in \cZ$ with distribution $V_x$. The goal
is to determine a stochastic map (a channel) $\Gamma: \cM
\rightarrow \cX$ such that for the composite channel 
$V^\prime = V\circ \Gamma$, with inputs from $\cM$ and 
outputs in $\cZ$,  it holds that 
\begin{itemize}
\item[(i)] For a  {uniformly distributed}
 input $M$
  of $\Gamma$, the
random variable $M$ is almost independent of the output $Z$ of $V^\prime$ 
observed by the eavesdropper; and 
\item[(ii)] $M$ can be determined
from the input $X$ of $V$. 
\end{itemize}
Note that in the
source model discussed in the previous section, the
distribution of $X$ is
fixed and a uniformly distributed $M$ is obtained as $F(X)$,
the output of a randomly chosen member $F$ of a UHF. In
contrast, in the channel model we fix the distribution of
$M$
and seek to design $\Gamma$ such that the two properties above
hold. Here, too, a constructive scheme can be obtained using
a UHF satisfying certain ``balanced'' conditions.
Specifically, we consider a UHF $\{f_s: \cX\rightarrow
\{0,1\}^k, s\in \cS\}$ satisfying the following 
 {balanced condition}:
For every seed $s\in \cS$ and $m \in \{0,1\}^k$,
\begin{align}
|\{x\in \cX \mid f_s(x) = m\}| = 2^{b}.
\nonumber
\end{align}
The  {condition above} says that for each member of the
UHF, the cardinality of each inverse-image set is the same.
We call a UHF satisfying the condition above a
 {$b$-balanced} UHF.  

A  {$b$-balanced} UHF can be used to design 
the aforementioned stochastic map $\Gamma:\cM \rightarrow \cX$
as follows:
For each $m\in \cM$, choose $X$ uniformly over
$f_S^{-1}(m)$, where the random seed $S$ is chosen uniformly
over $\cS$. The next result is a counterpart of the leftover
hash lemma for the channel model and shows that the requirement
(i) above holds if $b$ is
less than a threshold. The first instance of such a result 
appears in 
\cite[Section V]{HayMat10}. The weaker version below 
uses a different threshold which is often easier to 
evaluate. Specifically, the threshold in the lemma below is 
given by the {\it smooth max-information} of the channel, which is defined as follows:
Consider a subnormalized channel $V:\cX \rightarrow \cZ$
with a finite input alphabet $\cX$ 
and such that for each $x\in \cX$ the measure $V(\cdot \mid
x)$ on $\cZ$ has a density  
$\omega(z|x)$ with respect to a measure $\mu$ on $\cZ$. The
{\it max-information}  
of $V$ is given by
\[
\maxI V = \log \int \max_{x\in \cX} \omega(z|x)\, d\mu.
\]

For a subset $\cT$ of $\cX\times \cZ$, denote by 
$V_\cT$ the subnormalized channel corresponding to the
density
\begin{align}
\omega_\cT (z\mid x) = 
\begin{cases}
\omega(z|x), \quad &(x,z)\in \cT,
\\
0, \quad &\text{otherwise}.
\end{cases}
\label{e:smoothing_def}
\end{align}
The $\ep$-smooth max-information of $V$, $\maxIep V$, is given by the infimum
of $\maxI{V_\tau}$ over all sets $\cT \subset \cX\times \cZ$ such that
\begin{align}
V(\{z: (x,z)\in \cT\} \mid x) \geq 1-\ep, \quad \text{for
  all } x\in \cX.
\label{e:smoothing_property}
\end{align}
Note that the smoothing operation in the definition of
smooth max-information is different from the one used in defining {smooth
  max-entropy} above, but is similar to the definition of
smoothing in \cite{RenWol05}.

\begin{lemma}[{{\bf Leftover hash: Channel model}}]\label{l:leftover_hash_channel}
Given a channel $V: \cX\rightarrow Z$, with a finite input
set $\cX$ and arbitrary output set $\cZ$, and an 
 {$b$-balanced}  $k$-bit UHF
$\{f_s: s\in \cS\}$, suppose 
  that for each    {$m \in \cM=\{0,1\}^k$} 
 and $s\in \cS$ the input $X$ of $V$ is chosen
  uniformly over $f_s^{-1}(m)$. Then, for a
   {random variable $M$ distributed uniformly
    on $\cM$}    , 
\begin{align}
 \MI M {Z,S} \leq \frac{1}{\ln
  2}\cdot 2^{-(b - \maxIep V)} +  {\ep k}, \text{
   }
\label{e:inverseLHLbound}
\end{align}
where the seed $S$ is distributed uniformly over $\cS$.
\end{lemma}
\begin{proof}
Consider a $b$-balanced UHF $\{f_s, s\in \cS\}$.
We first prove the bound in \eqref{e:inverseLHLbound} for
the special case of $\ep=0$. To that end, note first that the
conditional density of $Z$ (w.r.t. $\mu)$ given $M=m$ and 
$S=s$ is given by 
\begin{align}
\frac{d \bPP{Z|M, S}}{d\mu}(z|m,s) &= \sum_{x\in
  f_s^{-1}(m)} \frac1{|f_s^{-1}(m)|}\cdot \denW Z x 
\nonumber
\\
&= 2^{-b}\sum_{x\in \cX}\indicator{f_s(x)=m} \denW z x,
\label{e:cond_density1}
\end{align}
where the equality is by definition of a $b$-balanced UHF.
Similarly, since $M$ and $S$ are independent
and $M$ is distributed uniformly over $\{0,1\}^k$, the
conditional density of $Z$ (w.r.t. $\mu)$ given 
$S=s$ is given by  
\begin{align}
\frac{d \bPP{Z|S}}{d\mu}(z|s) &=
2^{-b-k}\sum_{m\in \cM}\sum_{x\in
  \cX}\indicator{f_s(x)=m} \denW z x
\nonumber
\\
&=
2^{-b-k}\sum_{x\in
  \cX} \denW z x,
\label{e:cond_density2}
\end{align}
where we have used $\sum_{m\in \cM}\indicator{f_s(x)=m}=1$.
By \eqref{e:cond_density1} and
\eqref{e:cond_density2}, we get
\begin{align}
&\MI M {Z,S}
\\
&= \CMI MZS
\nonumber
\\
&= \bEE \log \frac{d\bPP{Z|M, S} }{d\bPP{Z|S}}
\nonumber
\\
&= \bEE \log \frac
{2^{-b}\sum_{x^\prime\in \cX}\indicator{f_S(x^\prime)=M} \denW Z {x^\prime}}
{2^{-b-k}\sum_{x^{\prime \prime}\in
  \cX} \denW Z {x^{\prime\prime}}}
\nonumber
\\
&= \bEE \log \frac
{2^k\sum_{x^\prime\in \cX}\indicator{f_S(x^\prime)=M} \denW Z {x^\prime}}
{\sum_{x^{\prime\prime}\in \cX} \denW Z {x^{\prime\prime}}}
\nonumber
\\
&= \frac{2^{-b-k}}{|\cS|}\int_{\cZ}\sum_{s\in \cS, m\in \cM,
x\in \cX} \Bigg[\indicator{f_s(x)=m} \denW z x \log \frac
{2^k\sum_{x^\prime\in \cX}\indicator{f_s(x^\prime)=m} \denW z {x^\prime}}
{\sum_{x^{\prime\prime}\in \cX} \denW z
  {x^{\prime\prime}}}\Bigg]\mu(dz).
\label{e:MI_bound0}
\end{align}
Since the summand inside
$\big[\cdot\big]$ is nonzero only for $m = f_s(x)$, we
can replace the term $\indicator{f_s(x^\prime)=m}$ with
$\indicator{f_s(x^\prime)=f_s(x)}$ to obtain 
\begin{align}
&\MI M {Z,S} 
\nonumber
\\
&= \frac{2^{-b-k}}{|\cS|}\int_{\cZ}\sum_{s\in \cS, m\in \cM,
x\in \cX} \indicator{f_s(x)=m} \denW z x 
 \log \frac
{2^k\sum_{x^\prime\in \cX}\indicator{f_s(x^\prime)=f_s(x)} \denW z {x^\prime}}
{\sum_{x^{\prime\prime}\in \cX} \denW z
  {x^{\prime\prime}}}\mu(dz)
\nonumber
\\
&= \frac{2^{-b-k}}{|\cS|}\int_{\cZ}\sum_{s\in \cS,
x\in \cX}  \denW z x \log \frac
{2^k\sum_{x^\prime\in \cX}\indicator{f_s(x^\prime)=f_s(x)} \denW z {x^\prime}}
{\sum_{x^{\prime\prime}\in \cX} \denW z
  {x^{\prime\prime}}}\mu(dz)
\nonumber
\\
&\leq 2^{-b-k}\int_{\cZ}\sum_{x\in \cX}  \denW z x  
\log \frac
{2^k\sum_{x^\prime\in \cX}|\cS|^{-1}\sum_{s\in\cS}\indicator{f_s(x^\prime)=f_s(x)} \denW z {x^\prime}}
{\sum_{x^{\prime\prime}\in \cX} \denW z
  {x^{\prime\prime}}}\mu(dz),
\label{e:MI_bound}
\end{align}
where the last inequality is by Jensen's inequality applied
to the $\log$ function. Furthermore, using the UHF property 
\eqref{e:UHF_property} for the UHF $\{f_s, s\in \cS\}$ we
have
\[
\frac 1{|\cS|}\sum_{s\in\cS}\indicator{f_s(x^\prime)=f_s(x)}
\leq 2^{-k}\indicator{x^\prime \neq x} + \indicator{x^\prime
  = x},
\]
which along with \eqref{e:MI_bound} gives
\begin{align}
&\MI M {Z,S} 
\nonumber
\\
&\leq 2^{-b-k}\int_{\cZ}\sum_{x\in \cX}  \denW z x
\log \frac
{2^k\sum_{x^\prime\in \cX}\left( 
2^{-k}\indicator{x^\prime \neq x} + \indicator{x^\prime
  = x}
\right)\denW z {x^\prime}}
{\sum_{x^{\prime\prime}\in \cX} \denW z
  {x^{\prime\prime}}}\mu(dz)
\nonumber
\\
&\leq 2^{-b-k}\int_{\cZ}\sum_{x\in \cX}  \denW z x  \log \frac
{\sum_{x^\prime\in \cX} \denW z {x^\prime}
+ 2^{k} \denW z x
}
{\sum_{x^{\prime\prime}\in \cX} \denW z
  {x^{\prime\prime}}}\mu(dz)
\nonumber
\\
&= 2^{-b-k}\int_{\cZ}\sum_{x\in \cX}  \denW z x \log \left( 1+ \frac
{ 2^{k} \denW z x}
{\sum_{x^{\prime\prime}\in \cX} \denW z {x^{\prime\prime}}}\right)\mu(dz)
\nonumber
\\
&\leq \frac{2^{-b}}{\ln 2}\int_{\cZ}\frac
{ \sum_{x\in \cX} {\denW z x}^2}
{\sum_{x^{\prime\prime}\in \cX} \denW z {x^{\prime\prime}}}\mu(dz),
\nonumber
\end{align}
where the previous inequality uses $\ln (1+x)\leq x$ for
all $x\geq 0$. Therefore, on observing that
\[
\frac
{ \sum_{x\in \cX} {\denW z x}^2}
{\sum_{x^{\prime\prime}\in \cX} \denW z {x^{\prime\prime}}}
\leq \max_{x}\denW z x = 2^{\maxI W},
\]
we get
\begin{align*}
\MI M {Z,S} \leq \frac1 {\ln 2}\cdot 2^{-\left(b-\maxI W\right)},
\end{align*}
which completes the proof for the case $\ep =0$.

Moving to the case $\ep>0$, consider a set $\cT \subset
\cX\times \cZ$ satisfying \eqref{e:smoothing_property}.
Note that by log-sum inequality
\begin{align}
&\sum_{x\in \cX}\indicator{f_s(x)=m} \denW z x \log \frac
{2^k\sum_{x^\prime\in \cX}\indicator{f_s(x^\prime)=m} \denW z {x^\prime}}
{\sum_{x^{\prime\prime}\in \cX} \denW z
  {x^{\prime\prime}}}
\nonumber
\\
&\leq
\sum_{x\in \cX: (x,z)\in \cT}\indicator{f_s(x)=m} \denW z x \log \frac
{2^k\sum_{x^\prime\in \cX: (x^\prime,z)\in \cT}\indicator{f_s(x^\prime)=m} \denW z {x^\prime}}
{\sum_{x^{\prime\prime}\in \cX: (x^{\prime\prime},z)\in \cT} \denW z
  {x^{\prime\prime}}}
+
\nonumber
\\
&\quad \sum_{x\in \cX: (x,z)\in \cT^c}\indicator{f_s(x)=m} \denW z x \log \frac
{2^k\sum_{x^\prime\in \cX: (x^\prime,z)\in \cT^c}\indicator{f_s(x^\prime)=m} \denW z {x^\prime}}
{\sum_{x^{\prime\prime}\in \cX: (x^{\prime\prime},z)\in \cT^c} \denW z
  {x^{\prime\prime}}}
\label{e:log-sum}
\end{align}
Thus, upon denoting the right-side of \eqref{e:MI_bound0} by
$g(V)$, \eqref{e:MI_bound0} and \eqref{e:log-sum} give
\[
\MI M {Z, S}\leq g(V_\cT) + g(V_{\cT^c}),
\]
where $V_\cT$ and $V_{\cT^c}$ are defined in
\eqref{e:smoothing_def}. Proceeding as in
the $\ep=0$ case, we get
\[
g(V_\cT) \leq \frac 1 {\ln 2}\cdot2^{-b+\maxI {V_\cT}}.
\]
Furthermore, using the simple bound
\[
\log \frac
{2^k\sum_{x^\prime\in \cX: (x^\prime,z)\in \cT^c}\indicator{f_s(x^\prime)=m} \denW z {x^\prime}}
{\sum_{x^{\prime\prime}\in \cX: (x^{\prime\prime},z)\in \cT^c} \denW z
  {x^{\prime\prime}}}
\leq k
\]
we get
\[
g(V_{\cT^c}) \leq \bPr{(X,Z)\in \cT^c}k \leq
\ep k,
\]
where the previous inequality uses the assumption that $\cT$
satisfies \eqref{e:smoothing_property}. It follows upon combining the
inequalities above that
\[
\MI M {Z, S}\leq \frac 1{\ln 2}\cdot2^{-\left(b-\maxI {V_\cT}\right)} + \ep k.
\]
The proof is completed using the definition of $\ep$-smooth
max-information upon optimizing $\maxI {V_\cT}$ over
sets $\cT$ that satisfy \eqref{e:smoothing_property}.
\end{proof}
Thus, the leakage  {$\MI M {Z,S}$} is small as long as $b$ is much smaller than
$\maxIep V$.  As in the case of source model, here, too, it
is of interest to determine the optimal leakage
exponent. Furthermore, it is of interest to derive bounds on
leakage for other measures such as the total variation
distance measure\footnote{In applying this bound to the case
  of wiretap channel, the channel $V$ will be chosen to be
  the concatenation of the legitimate transmission channel
  and an error correcting code for it.}; one instance of
such bound is available in \cite{Hayashi13ii} for the
special case when the channel $V$ is given by a
concatenation of a random code and another transmission
channel.

\subsection{Implementations}
An efficient implementation of a $k$-bit UHF for an $l$-bit
input can be obtained as follows \cite{CarWeg79, Kra94}: Let
$\{0,1\}^l$ correspond to the elements of
$GF\left(2^l\right)$ and let $\cS = \{0,1\}^l\setminus
\{{\bf 0}\}$. For $k \leq l$, define a mapping $f: \cS
\times \{0,1\}^l \rightarrow \{0,1\}^k$ as follows:
\begin{align}
f(s,x) &= (s \mult x)_k,
\nonumber
\end{align}
where $(x)_k$ selects the $k$ most significant bits of
$x$. It is easy to see that the family of mappings $\{f_s(x)
:= f(s,x), s \in \cS\}$ constitutes a UHF.
In fact,  {it is easy to see that this UHF is
  a $(l-k)$-balanced UHF}. Furthermore,
for $m \in \cM=\{0,1\}^k$, a uniform distribution on the
inverse-image set $f_s^{-1}(m)$ (required in the channel
version of the leftover hash lemma) can be computed
efficiently, too, using the mapping $\phi(s,m,R) =
s^{-1}\mult (m\conc R),$ where $R$ denotes $(l-k)$ uniform
random bits and $(m \conc R)$ denotes the concatenation of
$m$ and $R$. Note that $\phi$ is indeed the inverse of $f$
since $f(s, \phi(s, m, r)) = m$ for every $s, m, r$.

Note that in order to implement the aforementioned UHF $f_s$
(and its inverse $\phi$) efficiently, we require an
efficient implementation of multiplication and inversion in
$GF(2^l)$. One such efficient implementation was given in
\cite{Silverman99} for special values of $l$. Specifically,
since the polynomial
\[
\Phi(X) = X^{l} + X^{l-1} + ... + X +1
\]
is irreducible in $GF(2)[X]$ if and only if
\begin{enumerate}
\item $l+1$ is prime, and
\item $2$ is a primitive root modulo $l+1$, $i.e.$, the
  powers $1, 2, 2^2, ..., 2^l$ are distinct modulo $l+1$,
\end{enumerate}
for the values of $l$ satisfying the two conditions above,
$GF(2^l)$ can be embedded as a subring of polynomials modulo
$X^{l+1} - 1$. In this case, the multiplication of two
elements in $GF(2^l)$ is tantamount to multiplying the
corresponding polynomials modulo $X^{l+1} - 1$, which in
turn corresponds to the convolution of the two binary
vectors of length $l$. As is well-known, this convolution
can be realized using $O(l\log l)$ computations using FFT,
and also on hardware using a  {linear finite shift register
(LFSR)}  of length $l$.  Also, the 
inverse of elements of $GF(2^l)$, too, can be computed
efficiently following the algorithm outlined in
\cite[Section 2.5]{Silverman99}.

The main limitation of the construction above is that it is
feasible only for selected values of $l$ satisfying the two
conditions above. However, this is perhaps not a severe
limitation since, if Artin's conjecture holds, the number of
such $l$s is infinite and one can identify such an $l$ of a
practically relevant order by running a simple computer code
\footnote{A list of first $110$ such $l$'s is available on
  {\it http://oeis.org/A001122}.}.

An alternative construction, which circumvents the
aforementioned limitation on the input length $l$, entails
using a randomly chosen Toeplitz matrix. Specifically, for a
random seed $S$ consisting of $(l+k-1)$ bits, the hash
function $f_S: \{0,1\}^l \rightarrow \{0,1\}^k$ is given by
a $k\times l$ matrix $A$ with the first row and the first
column consisting of elements of $S$ and $A_{i,j} = A_{i-1,
  j-1}$ for $1<i\leq k$ and $1<j\leq l$.  It was shown in
\cite{MansourNT90} that the family of mappings $f_s(x) = Ax,
s\in \{0,1\}^{l+k-1},$ constitutes a $k$-bit UHF for inputs
of length $l$. Note that we can view the multiplication of
an $l$-length vector $x$ with a Toeplitz matrix $A$ as
multiplying the extended $(l+k-1)$-length vector
$\overline{x} = (x_1,...,x_l, 0,0,0 ...,0)$ with the
circulant extension of $A$ and taking the first $k$ entries
\cite{Kra94, HayT15}. Thus, we can efficiently implement
this UHF since multiplication with a circulant matrix is the
same as convolution, which in turn can be computed
efficiently using FFT.

A simple modification of the Toeplitz matrix based UHF above
was given in \cite{Hay11} for which the inverse-image set
can be efficiently computed as well.  In this modified
version, the random seed $S$ consisting of $(l-1)$ bits is
used first to form a $k\times (l-k)$ Toeplitz matrix $A$ as
before, but $f_S(x)$ is given by $[A, I]x$, where $I$ is the
$k$-dimensional identity matrix.  Clearly, the corresponding
family of mappings constitutes a $k$-bit UHF with input
length $l$.  Furthermore, for $m \in \{0,1\}^l$, a uniform
distribution on the inverse-image set $f_S^{-1}(m)$ can be
computed efficiently, too, using the mapping $\phi(S,m,R) =
(R\conc m - AR)$ where $R$ denotes $(l-k)$ uniform random
bits. Note that $\phi$ is indeed the inverse of $f$ since
\[
f_S(\phi(S, m, r)) = [A, I](r\conc m - Ar) = Ar + m - Ar =
m,
\]
for every $s, m, r$. However, this Toeplitz matrix based
construction does not satisfy the conditions for a balanced
UHF and, therefore, cannot be used in
Lemma~\ref{l:leftover_hash_channel}.  To wit, for a nonzero
vector $x\in\{0,1\}^l$ with the first $l-k$ entries $0$,
$f_s(x) = m $ holds for every $s$ if $m = - (x_{l-k+1}, ...,
x_l)$ and for no $s$ otherwise, thereby violating condition
(b) in the definition of a balanced UHF. Nevertheless, it
satisfies condition (a) and, by \cite[Section V]{HayMat10},
will satisfy Lemma~\ref{l:leftover_hash_channel} when we
restrict to a uniform random variable $M$.

It is also of interest to implement a UHF with as little
shared randomness $S$ as possible. See \cite{HayT15} for
constructions based on finite field arithmetic requiring the
best known lengths of the shared seed $S$. In particular,
see \cite[Table I]{HayT15} for a comparison of seed length
required by various implementations available in the
literature. Another concern in hardware implementation of
UHF is the power consumption. To this end, a variant of the
finite field arithmetic UHF proposed in \cite{BlackHKK99}
has been implemented as a low power CMOS circuit in
\cite{YukselKS04}.

For the remainder of this article, we shall assume that the
required UHF or balanced UHF is implemented using the finite
field arithmetic based construction described above and
depicted in Figure~\ref{f:efficient_UHF}.
\begin{figure}[H]
\centering \includegraphics[scale=.28]{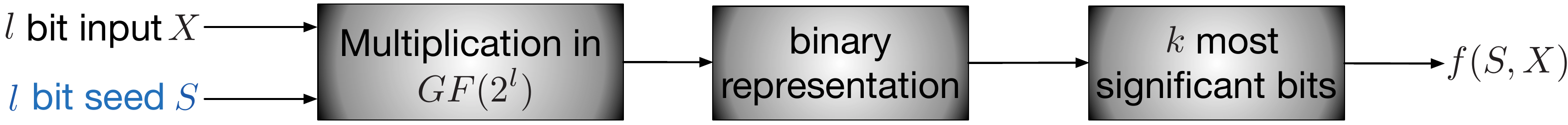}
\caption{An efficiently implementable UHF based on finite
  field arithmetic}
\label{f:efficient_UHF}
\end{figure}

\section{Practical secret key agreement schemes using UHF}
Extracting secret keys from correlated observations $X$ and
$Y$ has two obstacles.  First, although $X$ and $Y$ are
correlated they may not give rise to any shared randomness
for the two parties. In fact, a seminal result of G\'acs and
K\"orner \cite{GacKor73} says that, in general, correlation
cannot be converted into shared bits without
communication. Second, the shared bits that the parties can
generate by communicating may not be uniform or may not be
concealed from the eavesdropper with access to the
communication. All known secret key agreement schemes
circumvent these obstacles separately by first communicating
to agree on a shared randomness, a step referred to as {\it
  information reconciliation}, and then, extracting secret
keys from the generated shared randomness in the {\it
  privacy amplification} step. The choice of shared
randomness to generate and the tools for privacy
amplification vary across the literature. For instance, the
schemes in \cite{Mau93, AhlCsi93, Csi96, MauWol00, RenWol05,
  DodOstReySmi08, HayTyaWat14ii} recover $X$ as shared
randomness at both parties while that in \cite{CsiNar04,
  CsiNar08} recovers both $X$ and $Y$. Also, \cite{Tya13}
explores the role of the choice of shared randomness
established in the information reconciliation step in
reducing the amount of communication for secret key
agreement. For privacy amplification, \cite{Csi96, AhlCsi98,
  CsiNar04, CsiNar08} rely on the {\it balanced coloring
  lemma} which was introduced in \cite{AhlCsi98}.  On the
other hand, \cite{BenBraCreMau95, MauWol00, RenWol05, Ren05,
  HayTyaWat14i, HayTyaWat14ii} among several other works
rely on the leftover hash lemma.

A general construction in the context of biometric security
is given in \cite{DodOstReySmi08}. This construction is an
efficient implementation of the secret key agreement scheme
suggested in \cite{BenBraCreMau95} and \cite{RenWol05}, and
many special cases have appeared in implementation of PUFs;
see, for instance, \cite{GasClaDijDev02}. Also,
constructions based on  {low density parity
  check (LDPC)}  codes are given in
\cite{BlochTMM06} for a weaker notion of security, and the
ones on polar codes are given in \cite{RenesRS13, ChouBA13};
extensions to specific multiterminal models is considered in
\cite{YeNar12}.

We now describe a generic secret key agreement scheme that
can be implemented efficiently. For simplicity, assume that
$X = (X_1, ..., X_n)$ consists of $n$ independent, unbiased,
random bits and $Y = (Y_1,..., Y_n)$ is such that $(X_i,
Y_i)$ are mutually independent and each $Y_i$ is a possibly
flipped version of $X_i$, where flip occurs with probability
$\ep$. Therefore, for large $n$, the Hamming distance
between $X^n$ and $Y^n$ will be roughly $\tau = n\ep$.  In
fact, this scenario is typical, and it is common to process
and quantize the raw physical observations to extract
independent bits ($cf.$~\cite{YeMRSTM10, MaitiGS13}). The
extracted independent bits can be tested for independence
using standardized tests such as NIST SP-800-22-rev1a.  For
the purpose of this article, we shall assume that $n$
independent correlated bits $(X_i, Y_i)_{i=1}^n$ have been
extracted and have been distributed between the two parties.

The first component of our secret key agreement scheme is an
{\it error-correcting code} (ECC) that will facilitate a
compressed transmission of $X^n$ to $Y^n$. This classical
problem in distributed data compression was introduced by
Slepian and Wolf in \cite{SleWol73}, and several efficient
coding schemes accomplishing this are known. For instance,
\cite{LiverisXG02} gives an implementation based on LDPC
codes and \cite{KoradaU10} gives an implementation based on
polar codes. In fact, a simple implementation based on
linear ECC was suggested in \cite{Wyn74} and was used for
secret key agreement in \cite{YeNar12}; we review this
scheme here. Let $\cC$ be a linear ECC of length $n$ that
can be efficiently decoded and can correct up to $\tau$
errors. On observing $\bx$, the first party finds the coset
leader for $\bx$ in the standard array for the code
$\cC$. This can be implemented efficiently by using $\bx$ as
the input to an efficient decoder for $\cC$, noting the
decoded codeword $c_\bx$ and evaluating $e_\bx = \bx \oplus
c_\bx$.  This coset leader $e_{\bx}$ is communicated to the
second party over the public channel. The second party knows
$\by$ and computes $\by\oplus e_{\bx} = \bx \oplus e \oplus
e_{\bx} = c_{\bx}\oplus e$. Recall that $e$ has weight less
than $\tau$ with large probability, and therefore, $c_{\bx}$
can be recovered using the decoding algorithm for $\cC$. The
second party can recover $\bx$ as $c_\bx \oplus e_\bx$,
completing the information reconciliation step.

At this point, both parties agree on an $n$-bit vector
$X^n$, with a small probability of disagreement, $i.e.$, the
second party has an estimate $\widehat {X^n}$ of $X^n$ which
differs from $X^n$ with small probability of
error. Furthermore, a communication of, say, $r$ bits has
been revealed to the eavesdropper via the public channel. In
the privacy amplification step, the parties will use a UHF
to extract a secret key from shared bits $X^n$. In
particular, to use the UHF of Figure~\ref{f:efficient_UHF},
which can be implemented efficiently for input lengths $l$
such that $l+1$ is an odd prime and $2$ is a primitive root
modulo $l+1$, we find the largest such $l\leq n$ and use
just the first $l$ bits $(X_1,...,X_l)$.  In order to select
the range-size $k$, we first need to select a security
criterion and fix the desired security level under that
criteria. For instance, to attain a security of $\delta$
under the total variation distance, it follows
from\footnote{We apply Lemma~\ref{l:leftover_hash} with
  eavesdropper's side-information in the role of $Z_1$ and
  public communication in the role of $Z_2$.}
Lemma~\ref{l:leftover_hash} that $k = \lfloor
H^{\delta/2}_{\min}(\bPP{X^lZ^l}|Z^l) - r - 2\log(2/\delta)
\rfloor$ suffices, where $Z$ denotes the side-information of
the eavesdropper. Note that the security parameter $\delta$
is predecided and $r$ corresponds to maximum number of bits
that may be communicated in the information reconciliation
step.  Thus, to determine $k$, we only need to form an
estimate of the quantity
$H^{\delta/2}_{\min}(\bPP{X^lZ^l}|Z^l)$. For the case of IID
random variables $(X^l, Z^l)$ considered here, the smooth
conditional min-entropy
$H^{\delta/2}_{\min}(\bPP{X^lZ^l}|Z^l)$ can be approximated
by $lH(X|Z)$ (see \cite[Theorem 1]{HolensteinR11} for bounds
on approximation error at a fixed $l$).  The Shannon entropy
$H(X|Z)$ itself can be estimated by using $N =
\Theta(|\cX||\cZ|/\log |\cX||\cZ|)$ independent samples from
$(X,Z)$ \cite{ValVal13}.  If getting samples is expensive,
we can take recourse to an alternative form of the leftover
hash lemma where the threshold is determined by R\'enyi
entropy of order $2$ ($cf.$~\cite{ImpZuc89, BenBraCreMau95,
  Ren05}).  Specifically, using this form for the special
case of constant $Z$, we can find an appropriate value of
$k$ by estimating the R\'enyi entropy of order $2$ of $X$,
which requires only $\Theta(\sqrt{|\cX|})$ samples
\cite{AcharyaOST15}.

Once the value $k$ is determined, a secret key is extracted
by applying a $k$-bit UHF to $X^l$ and $\widehat{X^l}$ at
the first and the second party, respectively.  The overall
scheme discussed here is illustrated in
Figure~\ref{f:SK_scheme}.  The resulting secret key
agreement is capacity achieving if we use an optimal rate
Slepian-Wolf code in the information reconciliation
step. Note that the proposed scheme uses one-side
communication between the two parties, which can be strictly
suboptimal at finite blocklengths if interactive
communication is allowed \cite{HayTyaWat14ii}.
\begin{figure}[H]
\begin{subfigure}[b]{0.5\linewidth}
\centering \includegraphics[scale=.35]{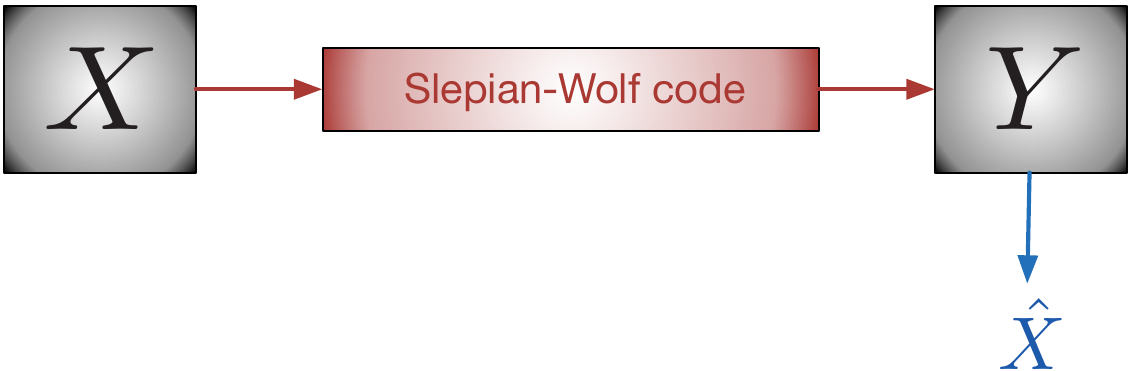}
\caption{Information reconciliation using a Slepian-Wolf
  code}
\end{subfigure}
\begin{subfigure}[b]{.5\linewidth}
\centering \includegraphics[scale=.35]{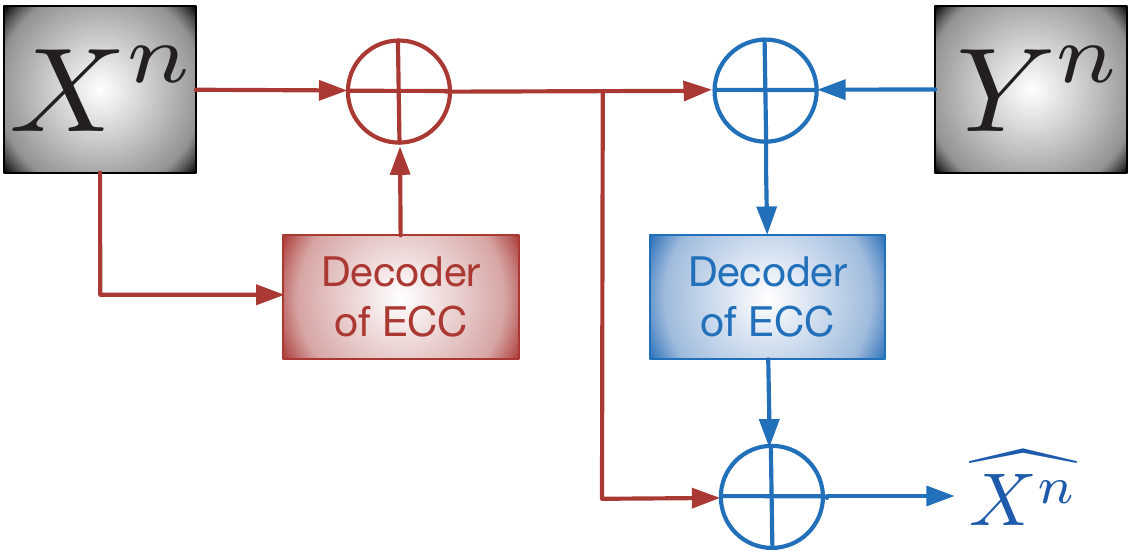}
\caption{Linear Slepian-Wolf code}
\end{subfigure}
\vspace*{0.3cm}

\begin{subfigure}[b]{\linewidth}
\centering \includegraphics[scale=.35]{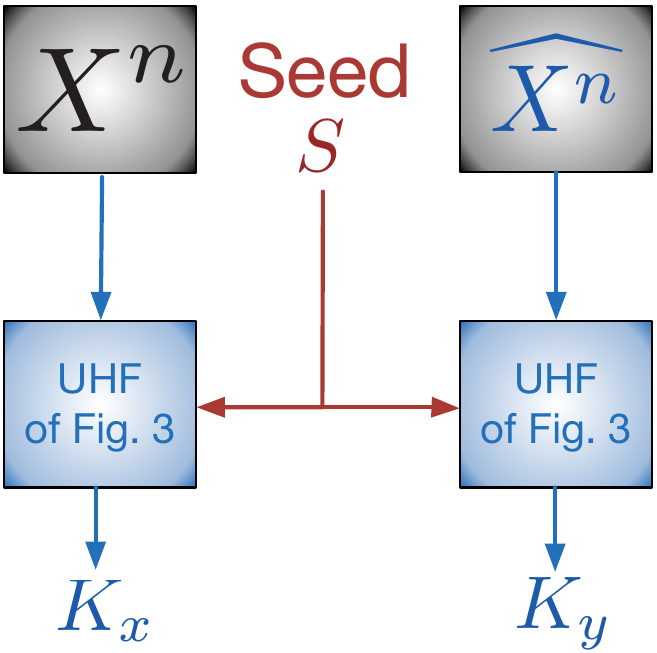}
\caption{Privacy amplification}
\end{subfigure}
\caption{An efficient scheme for secret key agreement}
\label{f:SK_scheme}
\end{figure}

\section{Practical (and modular) wiretap codes using UHF}
In order to present the main ideas underlying the
constructive wiretap coding schemes, we review briefly the
classical capacity-achieving, information-theoretic coding
schemes.

To construct an $(n,k)$ code for a wiretap channel (of rate
$k/n$), Wyner \cite{Wyn75ii} suggested to start with a code
$\cC$ of length $n$ and consider its partition $\cC =
\sqcup_{i=1}^{2^k} \cC_i$ such that
\begin{enumerate}
\item Each element of $\cC$ lies in the typical set
  $\cT_{[\bPP{}]}^n$ (for a definition of typical set, see
  \cite{CsiKor11});
\item $\cC$ is a ``good channel code'' for $T$ with small
  average probability of error;
\item each $\cC_i$ is a channel code for $T$ with average
  probability of error $\ep_i$, and the average of $\ep_i$
  with respect to $i$ is small.
\end{enumerate}
To encode a uniformly distributed message $M$, the channel
input $X^n$ is chosen uniformly over $\cC_M$.  It was shown
in \cite{Wyn75ii} that this scheme constitutes a valid
wiretap code. In fact, by selecting $\cC$ and its partition
randomly, we can attain the capacity of a degraded wiretap
channel.

Interestingly, while \cite{Wyn75ii} identified the general
properties that an ``ad-hoc'' channel code $\cC$ and the
corresponding partition $\sqcup_i \cC_i$ must satisfy to
yield a good wiretap code, the actual code construction in
\cite{Wyn75ii} entailed a joint selection of the code $\cC$
as well as the corresponding partition.  The construction in
\cite{CsiKor78} is of similar form and here, too, the
wiretap code is obtained by a joint selection of the random
channel code and its partition. The construction in
\cite{Csi96} (see, also, \cite{CsiKor11}), which attains the
wiretap capacity for a discrete, memoryless channel under
strong security, also starts with a random code $\cC$ and
partitions it using random binning\footnote{Strong security
  is shown by taking recourse to the {\it balanced coloring
    lemma}; see \cite{CsiKor11} for a detailed account.}.
The same holds for the scheme in \cite{Hay06} which relates
a randomly generated wiretap code to a channel resolvability
code~\cite{HanVer93}.

The information theoretic schemes above raise the following
question: Is it possible to obtain good wiretap codes by
starting with {\it any} good channel code for $T$ and
partitioning it appropriately? Or is the combined design
suggested in the schemes above necessary? In fact, most of
the constructive coding schemes proposed for a wiretap
channel follow the general template outlined above and
design wiretap codes by jointly selecting $\cC$ and its
partition, $i.e.$, the partition is selected, intrinsically,
based on the underlying code itself. For instance, the LDPC
codes based schemes in \cite{ThaDC+07} extend the {\it coset
  coding} scheme of \cite{OzarowW84} and select both the
partition and the overall code $\cC$ based on a specifically
designed parity check matrix (see
\cite[eqn. (21)]{ThaDC+07}); the polar codes based scheme in
\cite{MahV11} obtains the aforementioned partitioning, in
effect, by partitioning the polarized bits -- the polarized
bits that are ``good'' for the legitimate receiver yield the
overall code $\cC$ and the partition is obtained by fixing
the bits that are good only for the legitimate receiver, one
part for each fixed value of these bits (see, for instance,
\cite[eqn. (25)]{MahV11}); other polar coding schemes in
\cite{SosogluV13, RenesRS13, GolcuB14} have a similar form
except that the partitioning of polarized bits is more
involved -- a clear depiction of the partitioning of
polarized bits in these schemes is given in \cite[Figures
  1-4]{GolcuB14}; the same is true for the lattice codes
based scheme suggested for the Gausssian wiretap channel in
\cite{BelfioreOP11, LinLuzBelSte13, LinB14} where the
partition corresponds to appropriately selected cosets in
the transmission lattice. These schemes, while quite
important, will not be covered here in futher detail.  An
interested reader can see \cite{HarrisonABMB13} for a
review.

Thus, deployment of any of these schemes in place of
existing insecure channel codes will require a complete
redesign of the encoder and the decoder, which may not be
feasible. Recently, \cite{HayMat10, BelTes12} proposed a
{\it modular scheme} that starts with a good channel code
for $T$ and converts it into a good wiretap code by adding a
pre-processing layer based on UHFs\footnote{For a different,
  model of a wiretap channel, a coding scheme based on
  invertible extractors was given in
  \cite{CheraghchiDS11}.}. In fact, this modular scheme
appeared first in \cite{Hay11} for the special case when the
underlying channel code for $T$ is linear, and was shown to
achieve the capacity of a wiretap channel when both $T$ and
$W$ are additive, with strong security (based on the mutual
information criterion\footnote{It was extended to the total
  variation distance based security in \cite{Hay13}.}).  The
pre-processing layer of the proposed modular scheme is based
on UHFs and is shown to achieve the capacity of any
symmetric, degraded, discrete wiretap channel in
\cite{HayMat10, BelTes12, TalVar13} as well as that of a
Gaussian wiretap channel in \cite{TyagiV14} (see, also,
\cite[Appendix D]{HayM13}), both under strong security.  In
fact, when the underlying channel code for $T$ has a certain
linear structure, \cite{BelTes12} showed that this scheme
achieves the capacity of a symmetric, degraded, discrete
wiretap channel even under semantic security
(see, also, \cite{HayM13} for capacity results for the
modular scheme under different restrictions on the wiretap
channel and the underlying channel code for $T$).  It
remains unclear if such schemes can attain the capacities of
more general (including nondegraded) wiretap channels, as do
the schemes of \cite{RenesRS13, GolcuB14}, or how does their
overall performance compare with that of the schemes
mentioned above. Nevertheless, their ease of implementation
makes them a leading contender for deployment in practical
applications such as protection against side-channel attack
\cite{BringerCL12}.

In the remainder of this section, we review this modular
scheme. In the first subsection below, we begin by
presenting a {\it seeded wiretap coding scheme} where the
encoder and the decoder, additionally, have access to a
uniformly distributed random seed $S$. In the subsequent
subsection, this assumption of shared random seed will be
relaxed using the {\it seed recycling scheme} of
\cite{BelTesVar12, BelTesVar12ii}. Specifically, a seed $S$
is transmitted to the legitimate receiver over the first few
channel uses, and the same seed is re-used for multiple
instances of the seeded wiretap code. The security of this
combined scheme relying on seed recycling was established in
\cite{BelTesVar12, BelTesVar12ii} using a hybrid argument.

\subsection{Seeded wiretap codes}
To motivate the scheme, suppose that we transmit a message
$U$ by first encoding it using an ECC for $T$ and then the
legitimate receiver decodes $U$ as $\hat U$. Then, we are in
a similar situation as that in the secret key agreement of
Figure~\ref{f:SK_scheme} with $U$ and $\hat U$ corresponding
to the estimates of the reconciled information after the
first part of the scheme. We can extract a secret key $M$
from $U$ that remains concealed from the eavesdropper's
observations using a UHF, as in the privacy amplification
step of the scheme in Figure~\ref{f:SK_scheme}. However, in
the wiretap coding problem we are given a message $M$, and
we must generate $U$ from $M$ rather than the other way
around.  The main observation that leads to a wiretap coding
scheme is that if the extractor $F$ obtained by uniformly
choosing a mapping from a UHF is invertible, then we can
apply its inverse to the message $M$ to obtain $U$ and apply
the extractor itself to the decoded message $\hat U$,
thereby simulating the privacy amplification step in
Figure~\ref{f:SK_scheme} and ensuring security.

The key technical component required for formalizing this
idea is Lemma~\ref{l:leftover_hash_channel}, the channel
version of the leftover hash lemma. Specifically,
Lemma~\ref{l:leftover_hash_channel} shows that a balanced
UHF constitutes a stochastic transformation $\Gamma$ which
converts a given channel $V$ into a channel $V^\prime =
V\circ \Gamma$ with a different input alphabet $\cM$ but the
same output alphabet such that the input $m$ of $V^\prime$
remains secure from an observer of the output of $V^\prime$
and an observer of the random input of $V$ (output of
$\Gamma$) can determine $m$.

Suppose that we are given an ECC $\cC$ for the transmission
channel $T$ with encoder $e_0 : \{0,1\}^l \rightarrow
\cX^n$, where $\cX$ denotes the input of the wiretap
channel.  The code $\cC$ is assumed to facilitate a reliable
transmission of $l$-bit messages over $T^n$ with the maximum
probability of error less than $\mathrm{p}_{\tt e}$.  To
convert this code into an $(n,k)$ wiretap code, we add a
pre-processing layer to it consisting of a {$b$-balanced} $k$-bit UHF
$\{f_s, s\in \cS\}$ with input length $l$.  In order to send
a message $m \in \cM$, the pre-processing layer generates a
seed $S$ uniformly over $\cS$ and outputs a random binary
vector $U$ of length $l$ distributed uniformly over
$f_S^{-1}(m)$.  This vector $U$ is then encoded using $e_0$
and transmitted over $W$. In particular, we use the
efficiently invertible  {$(l-k)$-balanced} UHF
of Figure~\ref{f:efficient_UHF} for  {$\cM =
  \{0, 1\}^k$}. By our assumptions for the
code $\cC$, the random vector $U$ can be decoded at the
output of the transmission channel $T$ with probability of
error less than $\mathrm{p}_{\mathtt{e}}$. Thus, if the
random seed $S$ is available to the legitimate receiver, the
transmitted message $m\in\cM$, too, can be recovered with
probability of error less than $\mathrm{p}_{\mathtt{e}}$ by
applying $f_S$ to the decoded vector $\hat U$.  For the
security of this scheme, it follows from
Lemma~\ref{l:leftover_hash_channel}, applied with the
augmented channel $W_{n,\cC} = W^n \circ e_o$ in the role of
$V$,  {that for a uniformly distributed
  message $M$} 
\[
 \MI M {Z,S} \leq \frac{1}{\ln 2}\cdot2^{-(l -k - \maxIep {W_{n,
       \cC}})} + \ep k.
\]
\begin{figure}[H]
\begin{subfigure}[b]{\linewidth}
\centering \includegraphics[scale=.35]{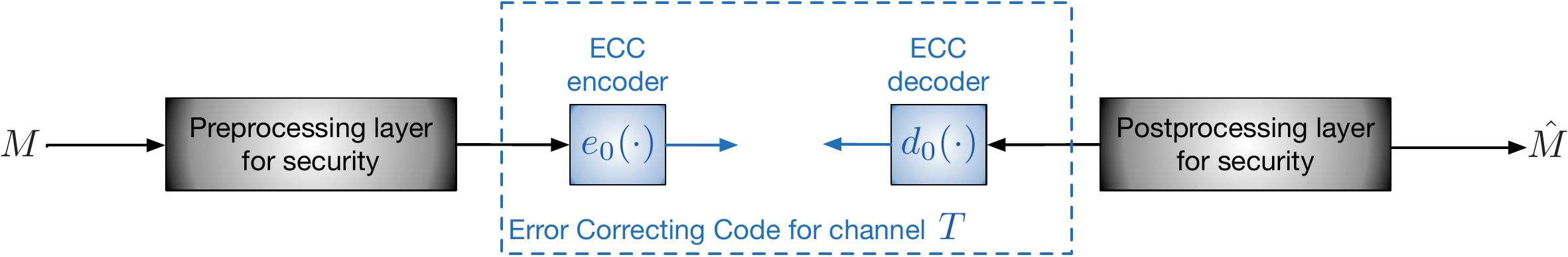}
\caption{Modular scheme for wiretap coding}
\end{subfigure}
\vspace*{0.3cm}

\begin{subfigure}[b]{0.5\linewidth}
\centering \includegraphics[scale=.35]{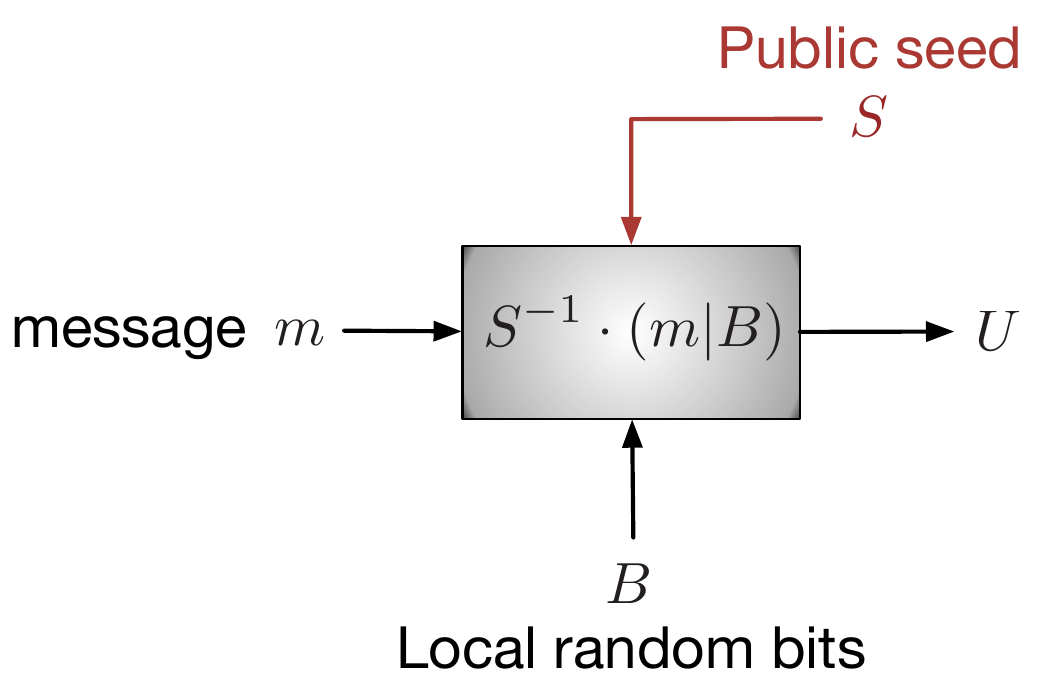}
\caption{The pre-processing layer}
\end{subfigure}
\begin{subfigure}[b]{0.5\linewidth}
\centering \includegraphics[scale=.35]{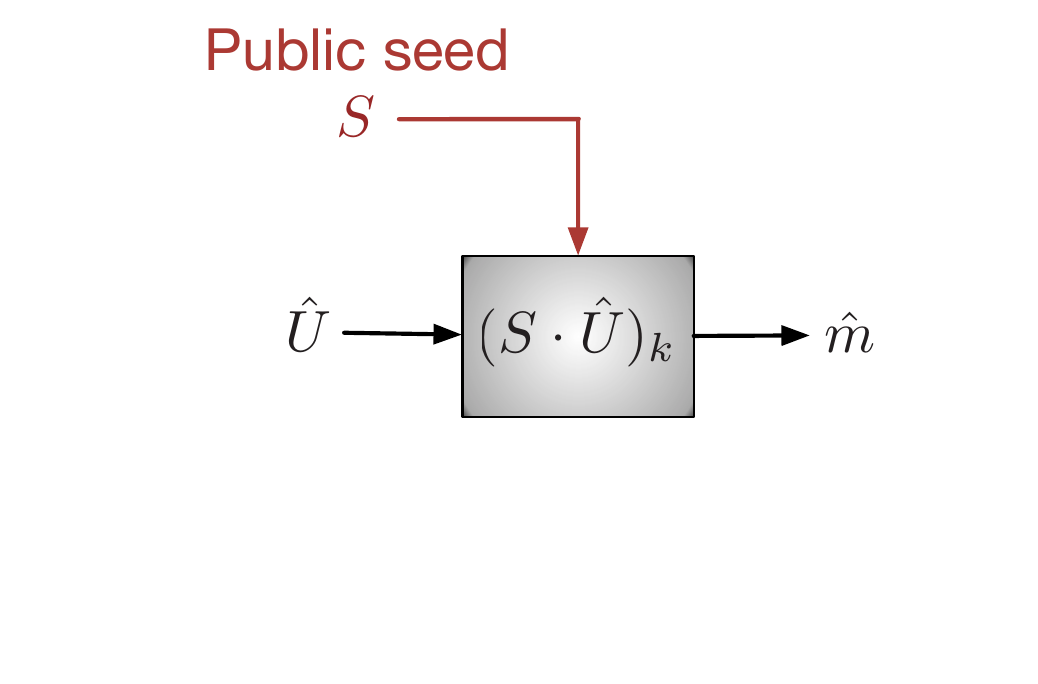}
\caption{The post-processing layer}
\end{subfigure}
\caption{A seeded coding scheme for the wiretap channel}
\label{f:wiretap_scheme}
\end{figure}
Therefore, the overall modular scheme, depicted in
Figure~\ref{f:wiretap_scheme}, constitutes a good $(n,k)$
wiretap code\footnote{To be precise, the proposed code with
  the choice of UHF in Figure~\ref{f:efficient_UHF} can send
  $2^k - 1$ messages because the all $0$ message is excluded
  from the message set $\cM$.} provided that $k$ is selected
appropriately to ensure small leakage $ \MI M
{Z,S}$. Specifically, suppose that the code $\cC$ is of rate
$R$, $i.e.$, $l = nR$.  We show in the Appendix
that there exists
a $c>0$ such that, for $\ep = 2^{-nc}$, $\maxIep {W_{n,
    \cC}}$ is asymptotically less than $nC_W$, where $C_W$
denotes the capacity of the channel $W$, both in the
case of a discrete memoryless channel (DMC) $W$ and in the
case of an additive white Gaussian noise (AWGN) channel $W$
with average input power constraints. Thus, upon choosing
\[
\frac k n = R^\prime < R - C_W,
\]
 {for a uniformly distributed $M$}, $\MI M {Z^n,S}$ vanishes to zero exponentially
rapidly in $n$.

 {In fact, for a symmetric, discrete channel $W$,
if the underlying ECC is linear and the balanced
$(l-k)$-balanced UHF of Figure~\ref{f:efficient_UHF}  is
used, it was shown in \cite{BelTesVar12,
  BelTesVar12ii} that strong security shown above implies
semantic security as well.} 

\subsection{Modular wiretap coding scheme based on seed recycling}
In the previous section, we established the security of  
our scheme assuming that a random seed $S$
was shared publically. We now show that this assumption is not required,
even for semantic security, using a seed recycling trick from 
\cite{BelTesVar12, BelTesVar12ii}.
We first use the 
legitimate channel $T$ to transmit
the seed $S$ reliably to the receiver in $nc$
channel uses, where the constant $c$ is chosen to ensure the 
recovery of $S$ at the receiver with probability of error less than 
$\mathrm{p}_{\tt e}$. Next, to compensate for the rate loss
due to the transmission of $S$,  
we use the same shared seed $S$ to send $t_n$ messages $M_1,
...,M_{t_n}$ using $t_n$ independent
implementations of the seeded wiretap
coding scheme of the previous subsection. The resulting
probability of error in 
transmitting the concatenated message 
$M^{t_n} = (M_1, ..., M_{t_n})$ in overall $N$ channel uses
is bounded above by $(t_n +1)\mathrm{p}_{\tt e}$. 
Also, by combining\footnote{As mentioned before, our
  security requirement is even stronger than the original {semantic
    security} requirement of \cite{BelTes12}, which can be
  shown for the combined scheme simply by using 
\cite[Lemma 4.2]{BelTes12}; \cite[Theorem 4.5,
  4.9]{BelTesVar12}
are required to move between the two notions of security.
}
\cite[Theorem 4.5, 4.9]{BelTesVar12}, \cite[Lemma 4.2]{BelTes12}, 
and the fact that 
$\MI M {Z^n,S}$ vanishes to $0$ exponentially rapidly in $n$,
it follows that 
$\max_{\bPP{M^{t_n}}}\MI {M^{t_n}} {Z^N}
\leq t_n 2^{-nc^\prime}$ for some constant $c^\prime>0$. 
The rate of the overall scheme is
given by
\[
\lim_{n\rightarrow\infty} \frac{t_nk}{(t_n+c)n},
\]
which equals $R^\prime$ as long as $t_n \rightarrow \infty$
as $n\rightarrow \infty$. Therefore, if we choose $t_n$ such
that this
condition is satisfied and both\footnote{The probability of error $\mathrm{p}_{{\tt e}}=\mathrm{p}_{{\tt e}, n}$
for the transmission code $\cC$ depends on $n$ and, in principle, can vanish to 
$0$ exponentially rapidly in $n$.} $(t_n +1)\mathrm{p}_{{\tt e}, n}$
and $t_n 2^{-nc^\prime}$ vanish to $0$, we get a
wiretap coding scheme satisfying semantic security 
of any rate $R^\prime < R-C_W$. Furthermore, if the
underlying  ECC $\cC$ can be implemented efficiently, so can
the combined scheme above.  {Note that the
  argument above is required to reduce the semantic security
  of an unseeded scheme to that of a seeded scheme. For the
  strong security criterion, a much simpler argument based
  on chain rule for mutual information suffices. Specifically,
consider a uniformly distributed message $(M_1, ..., M_{t_n})$.
Note that $M_i$, $1\leq i\leq t_n$ are IID uniform, which
further implies that for each $i$ the random variables
$(M_i, Z^n(i))$ are conditionally independent of
$(M_j,Z^n(j))_{j \neq i}$ given $S$.
Therefore,
\begin{align}
I(M_1, ..., M_{t_n} \wedge Z^N) &\leq
I(M_1, ..., M_{t_n} \wedge Z^N, S)
\nonumber 
\\
&= I(M_1, ..., M_{t_n} \wedge Z^n(1), ..., Z^n(t_n)|S)
\nonumber 
\\
&\leq \sum_{i=1}^{t_n} I(M_i \wedge Z^n(1), ..., Z^n(t_n)|S)
\nonumber 
\\
&= \sum_{i=1}^{t_n} I(M_i
\wedge Z^n(i)|S) 
\nonumber 
\\ & = t_n I(M_1 \wedge Z^n(1),S).
\label{e:no_seed_security}
\end{align} 
The security proof is completed by appropriately choosing $t_n\rightarrow \infty$
as above.

To summarize, the argument above allows us to convert any
efficiently 
implementable transmission code for $T$ of rate $R$ into 
a code of rate $R-C_W$ for the wiretap channel, with a vanishing
probability of error and under strong security. Furthermore, the conversion is done
simply by including an efficiently implementable
pre-processing layer based on a balanced UHF. Note that 
for the special case of a Gaussian wiretap channel or a
symmetric, degraded, discrete wiretap channel, the modular
scheme described above attains the wiretap capacity if the
underlying ECC $\cC$ achieves the capacity $C_T$ of the
transmission channel $T$
since, for these
cases, the wiretap capacity is given by
$C_T - C_W$ \cite{Che77}. In fact, for a discrete symmetric
wiretap channel, if the underlying
capacity achieving ECCC is linear and capacity achieving for
$T$, the modular scheme achieves the wiretap capacity even
under semantic security.

Recall that the balanced UHF of
Figure~\ref{f:wiretap_scheme} can be implemented efficiently
only for selected values of input length $l$. 
Thus, given an  ECC $\cC$ for $T$,
we simply use the largest $l$ less than
the input length of $\cC$ (in bits).
Also, the analysis above was asymptotic and cannot be
applied for a fixed $n$. For a fixed $n$, the output length
$k$ of the balanced UHF, and consequently the message
length, must be chosen to be appropriately smaller than
$l - \maxIep {W_{n, \cC}}$ to get the desired security
level. 
For this
purpose, it is required to estimate the quantity 
$\maxIep {W_{n, \cC}}$ for a given ECC $\cC$ and for
a sufficiently small $\ep$; however, there are no results to
report in this context yet. Furthermore, one might also wish
to compare the finite blocklength performance of this
scheme,
for different choices of ECC $\cC$, with the fundamental
lower bounds similar to those derived for the channel coding
problem in \cite{Hay09, PolPooVer10}. However, 
no such bounds are available.
In fact, even the strong converse for
a degraded wiretap channel was
proved only recently in \cite{HayTyaWat14iii}.

The coding scheme for the basic wiretap model above is a
stepping-stone for deriving schemes for more complicated
wiretap channel models such as the MIMO wiretap channel
considered in \cite{KhinaKK14}. It can be expected that,
based on the simple scheme above, schemes for other more
complicated physical-layer channel models will emerge. One
such extension, with a rather wide scope, appears in
\cite{HayM13}.

\section*{Appendix}
Consider a channel $W: \cX \rightarrow \cZ$ and 
an  encoder (for ECCs) $e_{0}:\{0,1\}^{l}\rightarrow
\{0,1\}^n$. Denote by $W_{e_0}:\{0,1\}^l\rightarrow \cZ^n$
the augmented channel $W^n\circ e_0$ where given an input
$v\in\{0,1\}^n$, with $x_i$ denoting the $i$th coordinate of
$e_0(v)$, the outputs $Z_i$ are independent and
distributed as $W(\cdot| x_i)$, $1\leq i \leq n$. 
In this section, we shall derive an asymptotic bound for
$I_{\max}^{\ep_n} (W_{e_{0}})$  for an exponentially small
$\ep_n$ and for a DMC $W$ with any
encoder $e_0$ as well as for an
AWGN channel $W$ with an encoder $e_{0}$ satisfying the
average power constraint $P$ with probability $1$. 

First, consider a DMC $W: \cX\rightarrow \cZ$. 
\begin{lemma}\label{l:bound_smooth_max_info_DMC}
For any encoder $e_{0}$ and a DMC $W:\cX\rightarrow \cZ$
with finite input and output alphabets $\cX$ and $\cY$,
respectively, there exists a constant $c>0$ such that for
$\ep_n = e^{-nc}$
\[
I_{\max}^{\ep_n} (W_{e_0}) \leq
n\max_{\bPP X}\MI X Z +o(n).
\]
\end{lemma}
\begin{proof}
We first prove the result for a constant composition code
where each codeword $e_0(v)$ is of a fixed type $P$,
$i.e.$, for $e_0$ such that each element $x\in \cX$
appears $nP(x)$ times in every codeword $x^n = e_0(v)$.

Denote by $\cT_{[W]}$ the set of sequences $(x^n, z^n)$ such that
$z^n$ is $W$-conditionally typical given $x^n$,
by $\cT_{[P,W]}$
the set of sequences $(x^n, z^n)\in \cT_{[W]}$ such that $x^n$ has
type $P$, 
and by $\cT_{[PW]}$ the projection of $\cT_{[P,W]}$ on
$\cZ^n$ (these notations are a slight deviation from those
used in \cite{CsiKor11}).  Then, using basic results from
the method of types (see \cite[Chapter 2]{CsiKor11}) for
each $(x^n, z^n)\in \cT_{[P, W]}$, it holds that
\[
\log W^n(z^n \mid x^n) \leq -nH(W\mid P) + o(n).
\]
Furthermore,
\[
\log |\cT_{[PW]}| \leq n H(PW)+o(n),
\]
where $PW$ denotes the output distribution for channel
$W$ when the input distribution is $P$.  Since there exists
a $c>0$ such that for all $v$
\[
\sum_{z^n:(e_0(v), z^n)\notin \cT_{[W]}}W^n(z^n \mid e_0(v))
\leq 2^{-nc},
\]
for the subnormalized channel $W_{e_0,\cT_{[W]}}$ defined by \eqref{e:smoothing_def},
we have 
\begin{align}
\maxIep {W_{e_0}}&\leq \maxI {W_{e_0,\cT_{[W]}}} \nonumber
\\ &= \log \sum_z \max_{v} W^n(z^n \mid
e_0(v))\indicator{(e_0(v),z^n)\in \cT_{[W]}} \nonumber
\\ &\leq -n H(W\mid P) + \log \sum_z \max_{v}
\indicator{(e_0(v),z^n)\in \cT_{[W]}} + o(n) \nonumber \\ &\leq
-n H(W\mid P) + \log |\cT_{[PW]}| + o(n) \nonumber \\ &\leq
nI(P; W) + o(n),
\label{e:constant_composition_Imax_bound}
\end{align}
where the last-but-one inequality uses the fact that 
$(e_0(v),z^n)\in \cT_{[W]}$ implies $z^n \in \cT_{[PW]}$
and so
\[
\max_{v}
\indicator{(e_0(v),z^n)\in \cT_{[W]}} \leq \indicator {z^n \in \cT_{[PW]}}
= |\cT_{[PW]}|.
\]
This completes the proof for a constant composition code.

Proceeding to the case of a general code, denote by $\cC_P$
the set of codewords $e_0(v)$ of type $P$.  As before, we
have
\begin{align*}
\maxIep {W_{e_0}}&\leq \maxI {W_{e_0,\cT_{[W]}}} \\ &= \log
\sum_z \max_{v} W^n(z^n \mid
e_0(v))\indicator{(e_0(v),z^n)\in \cT_{[W]}} \\ &\leq \log
\sum_P\sum_z \max_{x^n \in \cC_P} W^n(z^n \mid
x^n)\indicator{(x^n,z^n)\in \cT_{[W]}} \\ &\leq n\max_{\bPP
  X}I(\bPP X; W) + o(n),
\end{align*}
where the final inequality is obtained in the manner of
\eqref{e:constant_composition_Imax_bound} upon using the
fact that the number of types is polynomial in $n$ ($cf.$
\cite[Lemma 2.1]{CsiKor11}).
\end{proof}
Next, consider an AWGN channel $W: \mR\rightarrow \mR$,
$i.e.$, a channel such that for an  input $x \in \mR$ the
output $Z$ is 
distributed as $W(\cdot| x) = \cN(0, \sigma_W^2)$. 
Let  $e_0:\{0,1\}^l \rightarrow \{0,1\}^n$ be 
an encoder satisfying the average power constraint
\begin{align}
\frac 1n\|e_0(v)\|_2^2 \leq P, \quad \forall \, v\in
\{0,1\}^l.
\label{e:ECC_constraint}
\end{align}
The next result shows that the $\ep_n$-smooth
max-information for $W_{e_0}$ is bounded above by, roughly,
$n$ times the capacity of the AWGN $W$ with average input
power constraint $P$, for an exponentially
small $\ep_n$.
\begin{lemma}\label{l:bound_Imax_AWGN}
Let $W: \mR \rightarrow \mR$ be an AWGN channel with noise
variance $\sigma_W^2$, and let
$e_0~:~\{0,1\}^l~\rightarrow~\mR^n$ be an encoder
satisfying~\eqref{e:ECC_constraint}.
Then, denoting $\ep_n = e^{-n\delta^2/8}$, for the
combined channel $W_{e_0}$ it holds 
that
\[
I_{\max}^{\ep_n}(W_{e_0}) \leq \frac n 2\log \left(1+
 \frac P{\sigma_W^2}\right) + {n\delta \log e} +
o(n),
\]
for every $0<\delta$ sufficiently small.
\end{lemma}
\begin{proof}
Denote by $g(z)$ the standard normal density on $\mR^n$, by $\cZ_0$
the set $\{z^n: \|z^n\|_2^2\leq n(\sigma_W^2+P)(1+\delta)\}$, and by
$\cZ_{x^n}$ the set $\{z^n : \|z^n - x^n\|_2^2\geq
n\sigma_W^2(1-\delta)\}$. Further, denote 
\[
\cT = \{(x^n,
z^n): \|x^n\|_2^2> nP, z^n \in \mR^n\} \bigcup \{(x^n, z^n): \|x^n\|_2^2\leq
nP, z^n \in \cZ_{x^n}\cap \cZ_0\}. 
\]
Then, by the tail bounds for   
non-central $\chi^2$ RVs in \cite[Lemma 8.1]{Bir01} and for $\chi^2$
RVs (cf. \cite[Exercise 2.1.30]{AndGuiZei10}),  
we have  
\[
W^n(\{z^n: (x^n,z^n) \in \cT\}| x^n)\geq 1- \ep_n,
\] 
when $\delta$ is sufficiently small.
The following inequalities ensue:
\begin{align*}
I_{\max}^{\ep_n}(W_{e_0}) &\leq \maxI {W_{e_{0}, \cT}}
\\ &=\log \int_{\mR^n} \max_v g\left(\frac{z -
  e_0(v)}{\sigma_W}\right)\indicator{(e_0(v),z^n)\in
  \cT}dz \\ &\leq \log
\frac{e^{-\frac{n(1-\delta)}{2}}}{(2\pi\sigma_W^2)^\frac{n}{2}}
\int_{\mR^n} \max_v \indicator{(e_0(v),z^n)\in \cT}dz
\\ &\leq \log
\frac{e^{-\frac{n(1-\delta)}{2}}}{(2\pi\sigma_W^2)^\frac{n}{2}}\mathrm{vol}\left(\cZ_0\right),
\end{align*}
where the previous two inequalities hold by the definition of $\cT$ since 
$e_0(v)$ 
satisfies~\eqref{e:ECC_constraint} for all $v$. Denote by
$\cB_n(\rho)$ the sphere of radius $\rho$ in $\mR^n$ and by
$\nu_n(\rho)$ its volume, which can be approximated as
(cf. \cite{Wang05})
\begin{align}
\nu_n(\rho) = \frac{1}{\sqrt{n\pi}}\left(\frac{2\pi
  e}{n}\right)^{\frac{n}{2}}\rho^n\left(1 + O(n^{-1})\right).
\nonumber
\end{align}
Therefore, applying the volume formula above to 
$\rho_n = \sqrt{n(\sigma_W^2+P)(1+\delta)}$ and continuing 
with the foregoing
bounds for $I_{\max}^{\ep_n}(W_{e_0})$, we get
\begin{align*}
I_{\max}^{\ep_n}(W_{e_0}) &\leq \log
\frac{e^{-\frac{n(1-\delta)}{2}}}{(2\pi\sigma_W^2)^\frac{n}{2}}\nu_n(\rho_n)
\\
&=\log
\left[\frac{e^{\frac{n\delta}{2}}}{\sqrt{n\pi}}
\left(\frac{\rho_n^2}{n\sigma_W^2}\right)^\frac{n}{2}\left(1 +
O(n^{-1})\right) \right]
\\
&=\log
\left[\frac{e^{\frac{n\delta}{2}}}
{\sqrt{n\pi}}
\left[\left(1  + \frac P{\sigma_W^2}\right)(1+\delta)\right]^{\frac n 2}\left(1 +
O(n^{-1})\right) \right]
\\
&\leq\frac n 2\log \left(1
  + \frac P{\sigma_W^2}\right)
+ {n\delta \log e} + o(n),
\end{align*}
where we have used $\log (1+x) \leq x\log e$ in the last inequality.
\end{proof}

\bibliography{IEEEabrv,references}
\bibliographystyle{IEEEtranS}

\end{document}

%% file: preamble.tex
\usepackage{makecell}
\usepackage{bbm}
\usepackage{graphicx}
\usepackage[usenames,dvipsnames]{color}
\usepackage{epsfig}
\usepackage{amssymb}
\usepackage{amsmath}
\usepackage{amsthm}
\usepackage{latexsym}
\usepackage{setspace}
\usepackage{bbm}
\usepackage{flushend}
\usepackage{authblk}
\usepackage[top=1in, bottom=1in, left=1in, right=1in]{geometry}
\usepackage{float}
\usepackage{caption}
\usepackage{subcaption}
\usepackage{dsfont}

\newcommand{\dP}{\mathrm{P}}
\newcommand{\dQ}{\mathrm{Q}}
\newcommand{\bPP}[1]{{\dP_{#1}}}
\newcommand{\bQQ}[1]{{\dQ_{#1}}}
\newcommand{\bPr}[1]{{\mathrm{P}}\left(#1\right)}

\newcommand{\bP}[2]{\mathrm{P}_{#1}\left({#2}\right)}

\newcommand{\bEE}{{\mathbb{E}}}

\newcommand{\cB}{{\mathcal B}}
\newcommand{\cC}{{\mathcal C}}

\newcommand{\cF}{{\mathcal F}}

\newcommand{\cM}{{\mathcal M}}
\newcommand{\cN}{{\mathcal N}}

\newcommand{\cS}{{\mathcal S}}
\newcommand{\cT}{{\mathcal T}}

\newcommand{\cX}{{\mathcal X}}
\newcommand{\cY}{{\mathcal Y}}
\newcommand{\cZ}{{\mathcal Z}}

\newcommand{\bx}{\mathbf{x}}
\newcommand{\by}{\mathbf{y}}

\newcommand{\ep}{\epsilon}

\usepackage{color}

\newtheorem{theorem}{Theorem}

\newtheorem*{corollary*}{Corollary}

\newtheorem{lemma}[theorem]{Lemma}
\newtheorem*{lemma*}{Lemmas}

\theoremstyle{remark}
\newtheorem*{remark*}{Remark}
\newtheorem*{remarks*}{Remarks}

\theoremstyle{definition}

\newcommand{\MI}[2]{I\left(#1 \wedge #2\right)}
\newcommand{\CMI}[3]{I\left(#1 \wedge #2\mid #3\right)}

\newcommand{\indicator}[1]{{\mathds{1}\left(#1\right)}}

\newcommand{\mR}{\mathbb{R}}

\newcommand{\mult}{\cdot}
\newcommand{\conc}{,}
\newcommand{\maxI}[1]{I_{\max}\left(#1\right)}
\newcommand{\maxIep}[1]{I_{\max}^\ep\left(#1\right)}
\newcommand{\denW}[2]{\omega(#1 \mid #2)}
\newcommand{\wircap}{C}

\newcommand{\hl}[1]{}